\documentclass[reqno]{amsart}
\usepackage{amssymb}

\newtheorem{theorem}{Theorem}
\newtheorem{proposition}[theorem]{Proposition}
\newtheorem{lemma}[theorem]{Lemma}
\newtheorem{corollary}[theorem]{Corollary}

\theoremstyle{remark}
\newtheorem{remark}{Remark}

\numberwithin{equation}{section}

\begin{document}

\title[Bispectral dual for the quantum Toda chain]
{Bispectral dual difference equations for the quantum Toda chain with boundary perturbations}

\author{J.F.  van Diejen}

\address{
Instituto de Matem\'atica y F\'{\i}sica, Universidad de Talca,
Casilla 747, Talca, Chile}

\email{diejen@inst-mat.utalca.cl}

\author{E. Emsiz}

\address{
Facultad de Matem\'aticas, Pontificia Universidad Cat\'olica de Chile,
Casilla 306, Correo 22, Santiago, Chile}
\email{eemsiz@mat.uc.cl}


\thanks{This work was supported in part by the {\em Fondo Nacional de Desarrollo
Cient\'{\i}fico y Tecnol\'ogico (FONDECYT)} Grants  \# 1141114 and  \# 1170179.}

\date{August 2017}

\begin{abstract}
We show that hyperoctahedral Whittaker functions---diagonalizing an open quantum Toda chain with one-sided boundary potentials
of Morse type---satisfy a dual system of difference equations in the spectral variable.
This extends a well-known  bispectral duality between the nonperturbed open quantum Toda chain 
and a strong-coupling limit of the rational Macdonald-Ruijsenaars difference operators. It is manifest from the difference equations in question that  the hyperoctahedral Whittaker function is entire as a function of the spectral variable.
\end{abstract}

\maketitle

\section{Introduction}\label{sec1}
We consider hyperoctahedral Whittaker functions that arise as eigenfunctions for the Hamiltonian of an open quantum Toda chain with one-sided boundary potentials of Morse type:
\begin{align}\label{H}
\text{H}= -\frac{1}{2} \sum_{1\leq j\leq n} \frac{\partial^2}{\partial x_j^2}  
&+ e^{-x_1+x_2}+e^{-x_2+x_3}+\cdots +e^{-x_{n-1}+x_n}\\
&+a e^{-x_n}+b e^{-2x_n}. \nonumber
\end{align}
Here the independent variables $x_1,\ldots ,x_n$ represent one-dimensional particle positions and the constants $a, b $ denote two coupling parameters governing the perturbation of the chain at the boundary.
It is  known that the boundary potentials in question preserve the integrability of the Toda chain  \cite{skl:boundary,ino:finite} and that they deform
several integrable boundary perturbations stemming from the simple Lie algebras of classical type \cite{bog:perturbations,kos:quantization,ols-per:quantum,goo-wal:classical,sem:quantisation}.
Specifically, for $ab= 0$  the quantum Toda Hamiltonians associated with the Lie algebras of
type $A_{n-1}$ ($a= 0$, $b=  0$), type $B_n$ ($a\neq 0$, $b= 0$) and type $C_n$ ($a=  0$, $b\neq 0$) are recovered as degenerations.

Our main objective is to show that the hyperoctahedral Whittaker functions diagonalizing $\text{H}$ \eqref{H} satisfy a dual system of difference equations in the spectral variable, which
themselves can be interpreted as eigenvalue equations for a quantum integrable particle system. The commuting quantum integrals for this dual particle system turn out to be given by a strong-coupling limit of  rational Macdonald-Ruijsenaars type difference operators with hyperoctahedral symmetry introduced in \cite{die:integrability}.
This extends a corresponding bispectral duality \cite{dui-gru:differential,gru:bispectral} between the  nonperturbed open quantum Toda chain and a strong-coupling limit of the  (conventional) rational Macdonald-Ruijsenaars operators  \cite{kha-leb:integral,bab:equations,hal-rui:kernel,skl:bispectrality,koz:aspects,bor-cor:macdonald}.
At the level of classical mechanics, the latter bispectral duality for the (standard $A_{n-1}$ type) open Toda chain manifests itself in the fact that 
the action-angle transforms linearizing the corresponding classical-mechanical particle system and its dual are inverses of each other \cite{rui:relativistic,feh:action}.

For any  simple Lie algebra, the quantum Toda chain is diagonalized by an associated (class-one) Whittaker function
\cite{kos:quantization,goo-wal:classical,sem:quantisation,eti:whittaker}. These multivariate confluent hypergeometric functions have been studied extensively in the literature in diverse contexts, cf. e.g. 
\cite{jac:fonctions,has:whittaker,sta:explicit,sta:mellin,kha-leb:integral,ior-sha:wave,ish-sta:new,bau-oco:exponential,ger-leb-obl:new,rie:mirror,bor-cor:macdonald,cor-oco-sep-zyg:tropical,oco-sep-zyg:geometric,bru-bum-lic:whittaker} and references therein.
In relation to the quantum Toda chain, the Whittaker function arises e.g. via a connection formula that  both normalizes and symmetrizes the Harish-Chandra series solution of the eigenvalue problem (also known as the fundamental Whittaker function) \cite{has:whittaker,bau-oco:exponential}. 
It can moreover be
seen as a confluent limit \cite{shi:limit,osh-shi:heckman-opdam} of the Heckman-Opdam hypergeometric function (pertaining to the (reduced) root system of the underlying Lie algebra) \cite{hec-sch:harmonic,opd:lecture}.
With the aid of this kind of hypergeometric confluences,  dual
difference equations for the Heckman-Opdam hypergeometric function were recently seen to degenerate to corresponding difference equations for the Whittaker function \cite{die-ems:difference}.
By extending this scheme to the case of the Heckman-Opdam hypergeometric function associated with the
nonreduced root system of type $BC_n$, we arrive below at the hyperoctahedral Whittaker function diagonalizing $\text{H}$ \eqref{H} together with a corresponding system of dual difference equations in the spectral variable. As an application, the difference equations in question are employed
to infer that the dependence of the hyperoctahedral Whittaker function on the spectral variable is holomorphic everywhere.

The presentation is organized as follows. In Sections \ref{HC:sec}  and \ref{HWF:sec} the Harish-Chandra series and the hyperoctahedral Whittaker function for the quantum Toda Hamiltonian $\text{H}$ \eqref{H} are  constructed for generic values of the spectral variable. The associated
bispectral dual difference equations  are exhibited in Section \ref{bd:sec}, and in
Section \ref{ac:sec} these are subsequently used  to deduce that  the hyperoctahedral Whittaker function extends to an entire function  of the spectral variable.
Some of the more technical parts of our discusion are isolated from the main flow of the arguments and postponed towards the end.
Specifically, Section \ref{FW-function:prf} confirms
the convergence of the Harish-Chandra series for the fundamental Whittaker function, and Section \ref{regular-domain:prf} verifies that the function in question arises from the Harish-Chandra series for the hypergeometric equation of type $BC_n$ via a confluent limit of the kind considered by Shimeno and Oshima \cite{shi:limit,osh-shi:heckman-opdam}. In
Section \ref{HWD:prf}   this  same confluence is employed to retrieve the dual difference equations from their  hypergeometric counterparts in   \cite[Thm. 2]{die-ems:difference}.

\begin{remark}\label{BCn:rem} For $ab\neq 0$, the Morse potential at the boundary depends effectively  on a single nontrivial coupling parameter only.
Indeed, by translating the center of mass of the particle system ($x_j\to x_j+c$, $j=1,\ldots, n$), we may normalize  \emph{one} of either two coupling parameters $a$ or $b$ ($\neq 0$) to a fixed (nonzero) value of choice. Below we will assume that $b\neq 0$ and use the translational freedom to fix the strength of this coupling at
\begin{equation}
{b\equiv \frac{1}{8}}
\end{equation}
(while the parameter $a$ is allowed to vanish). Notice in this connection that in $\text{H}$ \eqref{H} the coupling strength of  the pair potential $e^{-x_k+x_{k+1}}$ ($k\in \{1,\ldots ,n-1\}$) was also conveniently normalized at unit value in a similar way, by exploiting the center-of-mass translational freedom
of the  particle cluster corresponding to the positions $x_1,\ldots ,x_k$. 
In all our formulas below we can in principle undo 
the above normalization of $b$ (by reverting the translation of the center of mass). With some care and assuming that $a\neq 0$, it is then straightforward to recuperate also the case of a vanishing coupling strength $b$ by performing the limit $b\to 0$.
\end{remark}

\section{Harish-Chandra series: fundamental Whittaker function}\label{HC:sec}
Following a classical approach going back to  Harish-Chandra---cf. e.g. Ref. \cite[Ch. IV.5]{hel:groups} and \cite[\S 4]{has:whittaker}---we  first construct a
power-series solution of the eigenvalue problem for $\text{H}$ \eqref{H}. 

Let us denote the unit vectors of the standard basis for $\mathbb{C}^n$ by  $e_1,\ldots ,e_n$ and let $\langle\cdot ,\cdot\rangle$ represent the natural bilinear scalar product
turning these into an orthonormal basis, i.e. for $\xi=(\xi_1,\ldots ,\xi_n)\in\mathbb{C}^n$ and $x=(x_1,\ldots ,x_n)\in\mathbb{C}^n$:
\begin{equation*}
\langle \xi ,x\rangle := \xi_1 x_1+\cdots +\xi_n x_n.
\end{equation*}
The Toda Laplacian with Morse perturbations is defined as
\begin{subequations}
\begin{equation}\label{L}
L_x:= \sum_{1\leq j\leq n} \frac{\partial^2}{\partial x_j^2}  -\sum_{\alpha\in S}    \text{a}_{\alpha}e^{-\langle \alpha ,x\rangle },
\end{equation}
where 
\begin{equation}\label{S}
S:=\{ e_1-e_2,e_2-e_3,\ldots ,e_{n-1}-e_n,e_n,2e_n\}
\end{equation} 
and (we have picked the normalization, cf. Remark \ref{BCn:rem} above)
\begin{equation}\label{A}
\text{a}_{e_1-e_2}=\cdots=\text{a}_{e_{n-1}-e_n}=2,\quad \text{a}_{e_n}=g,\quad \text{a}_{2e_n}=\frac{1}{4}.
\end{equation}
\end{subequations}
So $\text{H}=-\frac{1}{2}L_x$ with $a=\frac{g}{2}$ and $b=\frac{1}{8}$.

For $\nu=(\nu_1,\ldots ,\nu_n)\in\mathbb{Z}^n$, we write that
\begin{equation}\label{dominant}
\nu\geq 0 \Leftrightarrow \nu_1+\cdots +\nu_k\geq 0\ \text{for}\ k=1,\ldots ,n,
\end{equation}
and that $\nu >0$ if $\nu\geq 0$ and $\nu\neq 0$.
Given  a wave vector $\xi$ in the  dense domain of $\mathbb{C}^n$ of the form
\begin{equation}\label{U+}
\mathbb{U}^n :=\{ \xi\in\mathbb{C}^n\mid   \langle \nu-2\xi ,\nu\rangle\neq 0, \forall\nu>0 \} ,
\end{equation}
the Harish-Chandra series is now defined as the formal power series
\begin{subequations}
\begin{equation}\label{HCs:a}
\phi_\xi(x;g):=\sum_{ \nu\geq 0}  a_\nu (\xi ;g)  e^{\langle \xi -\nu ,x\rangle } ,
\end{equation}
with expansion coefficients determined by the recurrence
\begin{equation}\label{HCs:b}
a_\nu (\xi ;g) =  \frac{1}{  \langle \nu-2\xi ,\nu\rangle}   \sum_{\alpha\in S}  \text{a}_\alpha a_{\nu -\alpha} (\xi ;g)\qquad (\nu >0)
\end{equation}
and the initial conditions
\begin{equation}\label{HCs:c}
a_\nu (\xi ;g) := \begin{cases}   1&\text{if}\ \nu =0, \\ 0&\text{if}\ \nu\not\geq 0.\end{cases} 
\end{equation}
\end{subequations}
Notice that the uniqueness of the expansion coefficients $a_\nu (\xi ;g) $ \eqref{HCs:b}, \eqref{HCs:c} is guaranteed as the cone of integral wave vectors $\nu\geq 0$ is nonnegatively generated by $S$ \eqref{S}.

The following proposition confirms that the above Harish-Chandra series converges to an eigenfunction of $L_x$ \eqref{L}--\eqref{A}. Adopting standard terminology, we will refer to this eigenfunction as the \emph{fundamental Whittaker function}.

\begin{proposition}[Fundamental Whittaker Function]\label{FW-function:prp}
(i) The Harish-Chandra series $\phi_\xi(x;g)$ \eqref{HCs:a}--\eqref{HCs:c} converges absolutely and uniformly on compacts to a holomorphic function
of  $(\xi,x,g)\in\mathbb{U}^n\times\mathbb{C}^n\times \mathbb{C}$.

(ii) The (fundamental Whittaker) function in question provides an eigenfunction of the Toda Laplacian $L_x$ \eqref{L}--\eqref{A}:
\begin{subequations}
\begin{equation}\label{ef:eq}
L_x\phi_\xi (x;g) =\langle\xi,\xi\rangle \phi_\xi (x;g)
\end{equation}
that enjoys a plane-wave asymptotics of the form
\begin{equation}\label{ef:as}
\lim_{x\to +\infty} | \phi_{\xi }(x;g)- e^{\langle \xi, x\rangle} | =0 \quad\text{for}\quad \emph{Re}(\xi)=0,
\end{equation}
\end{subequations}
where the notation $x\to +\infty$ means that $x_k-x_{k+1}\to +\infty$ for $k=1,\ldots ,n$ (with the convention that $x_{n+1}\equiv 0$).
\end{proposition}

The proof of this proposition hinges on growth estimates for the expansion coefficients $a_\nu (\xi ;g)$ stemming from the recurrence relations that ensure the absolute convergence of the power series.
The reader is referred to Section \ref{FW-function:prf} below for the particular features of this (classical) argument in the present
setting.

In view of Proposition \ref{FW-function:prp}, it is immediate that---as a function of the spectral variable $\xi$---the fundamental Whittaker function
extends uniquely to a meromorphic function on $ \mathbb{C}^n$, with possible poles at hyperplanes belonging to the locally finite collection 
$\mathbb{C}^n\setminus\mathbb{U}^n$.
Most of these singularities  turn out to be removable. 

\begin{proposition}[Regularity Domain]\label{regular-domain:prp}
The fundamental Whittaker function  $\phi_\xi(x;g)  $ \eqref{HCs:a}--\eqref{HCs:c} extends uniquely to a holomorphic function
of  $(\xi,x,g)\in\mathbb{C}_{\emph{reg},+}^n\times\mathbb{C}^n\times \mathbb{C}$, where
\begin{equation}
\mathbb{C}_{\emph{reg},+}^n:=\{ \xi\in\mathbb{C}^n\mid 2\xi_j\not\in\mathbb{Z}_{> 0}, \,\xi_j\pm\xi_k \not\in\mathbb{Z}_{> 0}\, (j<k)\} .
\end{equation}
Moreover, as a (meromorphic) function of the spectral variable $\xi\in\mathbb{C}^n$ the function under consideration has  at most simple poles along the hyperplanes belonging to $\mathbb{C}^n\setminus \mathbb{C}_{\emph{reg},+}^n$.
\end{proposition}

A corresponding regularity result for the Harish-Chandra series solving the hypergeometric equation associated with root systems was proven by Opdam,
cf. \cite[Cor. 2.2,\, Cor. 2.10]{opd:root}, \cite[Prp. 4.2.5]{hec-sch:harmonic}, \cite[Lem. 6.5]{opd:lecture} and \cite[Thm. 1.2]{nar-pas-pus:asymptotics}.
In Section \ref{regular-domain:prf} we will deduce Proposition \ref{regular-domain:prp} from Opdam's result with the aid of a hypergeometric confluence in the spirit of \cite{shi:limit,osh-shi:heckman-opdam}.

\section{Connection formula: hyperoctahedral Whittaker function}\label{HWF:sec}
Upon normalizing with an appropriate $c$-function, the
hyperoctahedral Whittaker function is built from the fundamental Whittaker function through symmetrization with respect to the hyperoctahedral group of signed permutations (acting on the spectral variable).
For the quantum Toda systems associated with the simple Lie algebras, an analogous construction of the corresponding Whittaker functions was carried out in Hashizume's seminal paper \cite{has:whittaker} (cf. also \cite[Sec. 4]{bau-oco:exponential}). Specifically, our formulas provide a parameter-deformation ($ab=0\to ab\neq 0$)  that unifies the Hashizume wave functions  for the quantum Toda Hamiltonians of type $B_n$ and $C_n$ (cf. Remark \ref{BCn:rem} above).

We consider the following action of the {\em hyperoctahedral group} $W:=S_n\ltimes  \{1 ,-1\}^n $
 of {\em signed permutations} $w=(\sigma ,\epsilon)$ on $\mathbb{C}^n$:
\begin{equation}
\xi=(\xi_1,\ldots ,\xi_n)\stackrel{w}{\longrightarrow} (\epsilon_1 \xi_{\sigma_1},\ldots ,\epsilon_n \xi_{\sigma_n})=:w\xi ,
\end{equation}
where $\sigma = \left( \begin{matrix} 1& 2& \cdots & n \\
 \sigma_1&\sigma_2&\cdots & \sigma_n
 \end{matrix}\right) $ belongs to the symmetric group $S_n$ and $\epsilon = (\epsilon_1,\ldots,\epsilon_n)$ with $\epsilon_j\in \{ 1,-1\}$ (for $j=1,\ldots, n$).
 Let  $\mathbb{C}_{\text{reg}}^n$ denote the dense domain of $\mathbb{C}^n$ of the form
 \begin{equation}\label{Creg}
 \mathbb{C}_{\text{reg}}^n:=\{ \xi\in\mathbb{C}^n\mid 2\xi_j\not\in\mathbb{Z}, \,\xi_j\pm\xi_k \not\in\mathbb{Z} \, (j<k)\}  .
 \end{equation}
For $(\xi,x,g)\in\mathbb{C}_{\text{reg}}^n\times\mathbb{C}^n\times \mathbb{C}$,   the \emph{hyperoctahedral Whittaker function} is now defined via the following \emph{connection formula}:
\begin{subequations}
\begin{equation}\label{connection-formula}
\Phi_\xi(x;g)   :=   \sum_{w\in W}   C(w\xi;g)  \phi_{w\xi}(x;g) ,
\end{equation}
\begin{equation}\label{c-function}
C(\xi;g) := \prod_{1\leq j\leq n}    \frac{\Gamma (2\xi_j)}{\Gamma (\frac{1}{2} +g +\xi_j)}
\prod_{1\leq j<k\leq n}   \Gamma (\xi_j+\xi_k )\Gamma (\xi_j-\xi_k) ,
\end{equation}
\end{subequations}
where $\Gamma (\cdot )$ refers to the gamma function \cite[Ch. 5]{olv-loz-boi-cla:nist}.  Notice that the $c$-function $C(\xi;g)$ \eqref{c-function} is regular
for $(\xi,g)\in\mathbb{C}_{\text{reg},-}\times\mathbb{C}$,
where
\begin{equation}\label{C-}
\mathbb{C}_{\text{reg},-}^n:=\{ \xi\in\mathbb{C}^n\mid 2\xi_j\not\in\mathbb{Z}_{\leq 0}, \,\xi_j\pm\xi_k \not\in\mathbb{Z}_{\leq 0}\, (j<k)\} 
\end{equation}
(so $\mathbb{C}_{\text{reg}}^n = \mathbb{C}_{\text{reg},+}^n\cap \mathbb{C}_{\text{reg},-}^n $), whence $\Phi_\xi(x;g) $ is a holomorphic function
of  $(\xi,x,g)\in\mathbb{C}_{\text{reg}}^n\times\mathbb{C}^n\times \mathbb{C}$.   Moreover, since the simple poles  along the hyperplanes $\xi_j=0$ and $\xi_j\pm \xi_k=0$ ($1\leq j\neq k\leq n$) stemming from the $c$-functions are removable in the final expression due to the hyperoctahedral symmetry, the regularity of $\Phi_\xi (x;g)$ is manifest  (in particular) when $\text{Re}(\xi)=0$ (cf. also Theorem \ref{ac-wf:thm} below for a much stronger statement ensuring the regularity of the hyperoctahedral Whittaker function for all values of $\xi\in\mathbb{C}^n$).

\begin{remark}\label{Toda-wavefunction:rem}
It is immediate from Proposition \ref{FW-function:prp} that the hyperocthedral Whittaker function
solves the eigenvalue equation for the Toda Laplacian $L_x$ \eqref{L}--\eqref{A}:
\begin{subequations}
\begin{equation}
L_x\Phi_\xi (x;g)  =\langle\xi,\xi\rangle \Phi_\xi (x;g) .
\end{equation}
By construction, this particular solution is
$W$-invariant  in the spectral variable $\Phi_{w\xi }(x;g)=\Phi_\xi (x;g)$, $\forall w\in W$, and it enjoys the following plane-wave asymptotics determined by the $c$-function $C(\xi;g)$  \eqref{c-function}:
\begin{equation}
\lim_{x\to +\infty}  \Bigl|  \Phi_{\xi }(x;g)-\sum_{w\in W} C(w\xi;g) e^{\langle w\xi, x\rangle}\Bigr| =0\quad\text{for}\quad\text{Re}(\xi)=0.
\end{equation}
\end{subequations}
\end{remark}

\section{Bispectral duality: difference equations}\label{bd:sec}
The (class-one) Whittaker functions diagonalizing the quantum Toda chains associated with the simple Lie algebras satisfy an explicit system of difference equations in the spectral variable \cite[Thm. 3]{die-ems:difference}. The identities in question generalize previously known difference equations for the Whittaker function associated with the Lie algebra of type $A_{n-1}$ \cite{kha-leb:integral,bab:equations,skl:bispectrality,koz:aspects,bor-cor:macdonald}. Our main result consists of the following system of difference equations for the hyperoctahedral Whittaker function diagonalizing the Toda Laplacian with Morse perturbation $L_x$ \eqref{L}--\eqref{A}. The proof of these difference equations is relegated to Section \ref{HWD:prf} below.

\begin{theorem}[Dual Difference Equations]\label{HWD:thm}
For any $(\xi,x,g)\in \mathbb{C}_{\emph{reg}}^n\times\mathbb{C}^n\times\mathbb{C}$ and $\ell \in \{ 1,\ldots ,n\}$,
the hyperoctahedral Whittaker function $\Phi_\xi(x;g)$ \eqref{connection-formula}, \eqref{c-function} satisfies the difference equation
\begin{subequations}
\begin{equation}\label{CHDE}
\sum_{\substack{J\subset \{ 1,\ldots ,n\} ,\, 0\leq|J|\leq \ell\\
               \epsilon_j \in \{ 1,-1\} ,\; j\in J}}
\!\!\!\!\!\!\!\!\!
U_{J^c,\, \ell -|J|}(\xi;g)
V_{\epsilon J}(\xi;g)
\Phi_{\xi +e_{\epsilon J}} (x;g ) =    e^{x_1+\cdots +x_\ell }\Phi_{\xi} (x;g) ,
\end{equation}
where
\begin{align}\label{V}
V_{\epsilon J}(\xi ;g)&:=
\prod_{j\in J} 
\frac{(\frac{1}{2}+g+\epsilon_j\xi_j)}{2\xi_j(2\xi_j+\epsilon_j)}
\prod_{\substack{j\in J\\ k\not\in J}} 
  (\xi_j^2-\xi_{k}^2)^{-1}
\nonumber \\
&  \times
\prod_{\substack{j,j^\prime \in J\\ j<j^\prime}}
(\epsilon_j\xi_j+\epsilon_{j^\prime}\xi_{j^\prime})^{-1}
(1+\epsilon_j\xi_j+\epsilon_{j^\prime}\xi_{j^\prime})^{-1}
\end{align}
and
\begin{align}\label{U}
U_{K,p}(\xi;g):= (-1)^{p(p+1)/2}
\sum_{\stackrel{I\subset K,\, |I|=p}
               {\epsilon_i  \in \{ 1,-1\} ,\; i\in I }}
&\Biggl( \prod_{i\in I} 
\frac{\frac{1}{2}+g+\epsilon_i\xi_i}{2\xi_i(2\xi_i+\epsilon_i)} \prod_{\substack{i\in I\\ k\in K\setminus I}} 
(\xi_i^2-\xi_{k}^2)^{-1}
 \\
&  \times
\prod_{\substack{i,i^\prime \in I\\ i<i^\prime}}
(\epsilon_i\xi_i+\epsilon_{i^\prime}\xi_{i^\prime})^{-1}
(1+\epsilon_i\xi_i+\epsilon_{i^\prime}\xi_{i^\prime})^{-1} \Biggr) . \nonumber \end{align}
\end{subequations}
Here $|J|$ denotes the cardinality of $J\subset\{ 1,\ldots, n\}$,  $J^c:=\{ 1,\ldots, n\}\setminus J$, and we have employed the compact notation
$e_{\epsilon J} := \sum_{j\in J} \epsilon_j e_j $.
\end{theorem}

For $\ell =1$  the difference equation in Theorem \ref{HWD:thm} boils down to the following identity
\begin{subequations}
\begin{equation}\label{HWD1:a}
\sum_{\substack{ 1\leq j\leq n\\  \epsilon \in \{ 1,-1\} }}  V_j(\epsilon \xi;g) \Bigl( \Phi_{\xi +\epsilon e_j}(x;g)-\Phi_\xi (x;g)\Bigl) 
 = e^{x_1} \Phi_\xi (x;g) ,
\end{equation}
where
\begin{equation}\label{HWD1:b}
V_j(\xi;g)=\frac{\frac{1}{2}+g+\xi_j}{2\xi_j(2\xi_j+1)}
\prod_{\substack{1\leq k\leq n\\ k\neq j}}  (\xi_j^2-\xi_{k}^2)^{-1} .
\end{equation}
\end{subequations}

\begin{remark}
Theorem \ref{HWD:thm} reveals that $\Phi_\xi(x;g)$ \eqref{connection-formula}, \eqref{c-function} constitutes a joint eigenfunction for a quantum integrable system of discrete difference operators acting on meromorphic functions of the spectral variable $\xi$:
\begin{equation}\label{Dl}
D_{\xi,\ell} := \sum_{\substack{J\subset \{ 1,\ldots ,n\} ,\, 0\leq|J|\leq \ell\\
               \epsilon_ j\in \{ 1,-1\} ,\; j\in J}}
\!\!\!\!\!\!\!\!\!
U_{J^c,\, \ell -|J|}(\xi;g)
V_{\epsilon J}(\xi;g)  T_{\epsilon J,\xi}\quad (\ell =1,\ldots ,n),
\end{equation}
where $(T_{\epsilon J,\xi} f)(\xi):= f(\xi +e_{\epsilon J})$ for $f:\mathbb{C}^n\to \mathbb{C}$.
Indeed, the  difference  operators $D_{\xi,1},\ldots ,D_{\xi,n}$ \eqref{Dl} amount to a strong-coupling limit of the
rational commuting hyperoctahedral Macdonald-Ruijsenaars type operators introduced in \cite[Sec. II]{die:integrability} (cf. also \cite[Eq. (2.10)]{die:difference}). This identifies the latter integrable quantum system as the bispectral dual \cite{gru:bispectral} of the quantum Toda system with Morse term $L_x$ \eqref{L}--\eqref{A}, via the following (bispectral)
extension of the eigenvalue equation in Remark \ref{Toda-wavefunction:rem}:
\begin{equation}
L_x\Phi_\xi (x;g) =\langle \xi ,\xi\rangle \Phi_\xi (x;g), \quad D_{\xi,\ell}\Phi_\xi (x;g) = e^{x_1+\cdots +x_\ell } \Phi_\xi (x;g) 
\end{equation}
($\ell=1,\ldots ,n$).
\end{remark}

\begin{remark} It is instructive to detail the above formulas somewhat further in the classical situation of a single particle ($n=1$).
It is then well-known---cf.  e.g.  Ref. \cite{lag:schrodinger} and references therein---that  for the corresponding Schr\"odinger eigenvalue problem  on the line with a Morse  potential
\begin{equation}
 L_x\phi=\xi^2\phi \quad \text{with}\quad  L_x=\frac{\text{d}^2}{\text{d}x^2}-ge^{-x}-\frac{1}{4}e^{-2x} ,
 \end{equation}
the unique solution characterized by a normalized plane-wave asymptotics of the form   $\lim_{x\to +\infty}    | \phi  (x)-e^{\xi x}|= 0$ for $\text{Re}(\xi)=0$ can be conveniently expressed explicitly in terms of Whittaker's (fundamental) $M$-function \cite[13.14.2]{olv-loz-boi-cla:nist}:
\begin{equation}\label{Mf}
\phi_\xi (x;g) = e^{\frac{x}{2}} M_{-g,-\xi} (e^{-x})=  e^{\xi x } e^{-\frac{1}{2}e^{-x}}  {}_1F_1 (\frac{1}{2}+g-\xi,1-2\xi ;e^{-x}) 
\end{equation}
($2\xi\not\in\mathbb{Z}_{>0}$). From the connection formula \eqref{connection-formula}, \eqref{c-function}, it is now manifest
that  for $n=1$ our hyperoctahedral Whittaker function reduces essentially to Whittaker's (class-one) $W$-function \cite[13.14.33]{olv-loz-boi-cla:nist}:
\begin{equation} \label{Wf}
\Phi_\xi (x;g) 
=\frac{\Gamma (2\xi)}{\Gamma (\frac{1}{2}+g+\xi)}  \phi_\xi (x;g)+ \frac{\Gamma (-2\xi)}{\Gamma (\frac{1}{2}+g-\xi)}  \phi_{-\xi }(x;g) \\
= e^{\frac{x}{2}}  W_{-g,\xi}(e^{-x}) 
\end{equation}
($2\xi\not\in\mathbb{Z}$). In the present univariate case, Theorem \ref{HWD:thm} specializes to the following elementary difference equation for this 
celebrated Whittaker function (cf. Eqs. \eqref{HWD1:a}, \eqref{HWD1:b}):
\begin{align}
&\frac{\frac{1}{2}+g+\xi}{2\xi (2\xi+1)} \Bigl( \Phi_{\xi +1}(x;g) - \Phi_\xi (x;g)\Bigr) +   \label{de}  \\
&\frac{(\frac{1}{2}+g-\xi)}{2\xi (2\xi-1)} \Bigl( \Phi_{\xi -1}(x;g)-\Phi_\xi (x;g)\Bigr) =
e^{x} \Phi_\xi (x;g) . \nonumber 
\end{align}
At $g=0$, the latter identity recovers a classical recurrence relation for Macdonald's modified Bessel function
$\Phi_\xi (x,0)= e^{\frac{x}{2}}  W_{0,\xi}(e^{-x}) = \frac{1}{\sqrt{\pi}}   K_\xi (\frac{1}{2}e^{-x})$ \cite[10.29.1, 13.18.9]{olv-loz-boi-cla:nist}:  \begin{equation}
\frac{1}{4\xi}    \Bigl( \Phi_{\xi +1}(x;0) - \Phi_{\xi-1} (x;0)\Bigr)=e^x \Phi_\xi (x;0) .
\end{equation}
For general $g$, the difference equation in Eq.  \eqref{de} can be retrieved by combining
the recurrence relations in  \cite[13.15.10, 13.15.12]{olv-loz-boi-cla:nist} (we thank T.H. Koornwinder for pointing this out).
\end{remark}

\section{Analytic continuation in the spectral variable}\label{ac:sec}
In this section we employ the difference equations of Theorem \ref{HWD:thm} to extend the hyperoctahedral Whittaker function analytically in the spectral variable to an entire function of $(\xi,x,g)\in\mathbb{C}^n\times\mathbb{C}^n\times\mathbb{C}$.

\subsection{Regularity}
Specifically, by inspecting the singularities of the difference equations we will read-off the following regularity  of $\Phi_\xi(x;g)$ in the spectral variable.

\begin{theorem}[Analyticity of the Hyperoctahedral Whittaker Function]\label{ac-wf:thm} 
The singularities of the
hyperoctahedral Whittaker function in the spectral variable $\xi\in\mathbb{C}^n$ along the hyperplanes of $\mathbb{C}^n\setminus\mathbb{C}^n_{\emph{reg}}$ are removable, i.e. by Hartogs' theorem $\Phi_\xi (x;g)$ \eqref{connection-formula}, \eqref{c-function} extends (uniquely) to an entire function of $(\xi,x,g)\in\mathbb{C}^n\times\mathbb{C}^n\times\mathbb{C}$.
\end{theorem}

This analyticity of the hyperoctahedral Whittaker function  implies in turn that the difference equations themselves extend as holomorphic identities on $\mathbb{C}^n\times\mathbb{C}^n\times\mathbb{C}$.

\begin{corollary}[Analyticity of the Dual Difference Equations]\label{ac-deq:cor} 
The difference equations of Theorem \ref{HWD:thm} extend (uniquely) to identities between entire functions of the variables $(\xi,x,g)\in\mathbb{C}^n\times\mathbb{C}^n\times\mathbb{C}$.
\end{corollary}

\subsection{Residue calculus}
While for any $\ell\in \{ 1,\ldots ,n\}$ the analyticity of both sides of the difference equation in
Corollary \ref{ac-deq:cor} is an immediate consequence of the asserted analyticity of  $\Phi_\xi (x;g)$ \eqref{connection-formula}, \eqref{c-function}  in Theorem \ref{ac-wf:thm}, the verification of the latter regularity requires
a detailed study of the residues of the hyperoctahedral Whittaker function in $\xi$ along the hyperplanes  of $\mathbb{C}^n\setminus\mathbb{C}^n_{\emph{reg}}$. In view of the hyperoctahedral symmetry, it is sufficient to infer that
the residues of $\Phi_\xi (x;g)$  along the hyperplanes $2\xi_1\in\mathbb{Z}_{>0}$ and $\xi_1+\xi_2\in\mathbb{Z}_{>0}$ vanish.

Specifically, for a fixed integer  $m$ and a meromorphic function $f:\mathbb{C}^n\to \mathbb{C}$ with at most simple poles along the hyperplanes
\begin{equation*}
H_{1;m}:=\{ \xi\in\mathbb{C}^n\mid  2\xi_1= m \} \quad\text{and}\quad  H_{12;m}:=\{ \xi\in\mathbb{C}^n\mid  \xi_1+\xi_2 = m \} ,
\end{equation*}
the residues of $f$ along these hyperplanes are given by the meromorphic functions
\begin{align*}
\text{Res}_{1;m} f (\xi) &:= \lim_{2\xi_1\to m}   (2\xi_1-m) f(\xi) , \\
\text{Res}_{12;m} f (\xi) &:= \lim_{\xi_1+\xi_2\to m}   (\xi_1+\xi_2-m) f(\xi)
\end{align*}
on $H_{1;m}$ and $H_{12;m}$, respectively.
In other words, $\text{Res}_{1;m}f(\xi)$ and $\text{Res}_{12;m}f(\xi)$ are the restrictions
$[(2\xi_1 -m) f(\xi)]_{2\xi_1=m}$ and $[(\xi_1+\xi_2 -m) f(\xi)]_{\xi_1+\xi_2=m}$
of the meromorphic functions $(2\xi_1 -m) f(\xi)$ to $H_{1;m}$   and  $(\xi_1+\xi_2 -m) f(\xi)$ to $H_{12;m}$, respectively.
We will now check that for $m\in\mathbb{Z}_{>0}$ both $\text{Res}_{1;m}\Phi_\xi(x;g)$ and $\text{Res}_{12;m}\Phi_\xi(x;g)$ vanish by induction in $m$, assuming that $\text{Res}_{1;k}\Phi_\xi(x;g)=0$ and $\text{Res}_{12;k}\Phi_\xi(x;g)=0$ for $0\leq k <m$.  Recall that for $m=1$ this induction hypothesis is fulfilled trivially by virtue of the hyperoctahedral symmetry.

\subsubsection{Verification that $\emph{Res}_{1;m}\Phi_\xi(x;g)=0$ for $m\in\mathbb{Z}_{>0}$}

We start from the simplest difference equation in Theorem \ref{HWD:thm} corresponding to $\ell=1$ (cf. Eqs. \eqref{HWD1:a}, \eqref{HWD1:b}):
\begin{align}\label{diffeqn1}
V_1(\xi)(\Phi_{\xi+e_1} -\Phi_\xi) &+V_{-1}(\xi)(\Phi_{\xi-e_1} -\Phi_\xi) \\
& +\sum_{2 \le j \le n,\epsilon\in\{1,-1\}}V_{\epsilon j}(\xi)(\Phi_{\xi+\epsilon e_j} -\Phi_\xi) 
 = e^{x_1}\Phi_\xi \nonumber
\end{align}
(where  the dependence on $x$ and $g$ is suppressed and we have written $V_{\epsilon j}(\xi)$ for $V_j(\epsilon \xi)$).

If \underline{$m=1$}, then multiplication of Eq.  \eqref{diffeqn1} by $(2\xi_1+1)^2$ and performing the limit 
$2\xi_1 +1\to 0$ yields the identity
$$
\text{Res}_{1;-1}V_1(\xi) \,   \text{Res}_{1;-1}   \Phi_{\xi+e_1} =0 .
$$
Since
$$
\text{Res}_{1;-1}V_1(\xi) =   -g
\prod_{1< k\leq n}  ({\textstyle \frac{1}{4}}-\xi_{k}^2)^{-1}       \not\equiv 0 ,
$$
it follows that $ \text{Res}_{1;-1}   \Phi_{\xi+e_1}(x,g) \equiv 0$ for $(\xi,x,g)\in H_{1;-1}\times\mathbb{C}^n\times\mathbb{C}$, and thus
$\text{Res}_{1;1}   \Phi_{\xi} (x;g) \equiv 0$  on $H_{1;1}\times\mathbb{C}^n\times\mathbb{C}$.

If \underline{$m=2$}, then multiplication of Eq. \eqref{diffeqn1} by $4\xi_1^2$ and performing the limit 
$2\xi_1\to 0$ yields that
\begin{equation*}
\text{Res}_{1;0}V_1(\xi) \text{Res}_{1;0} \Phi_{\xi+e_1} +
\text{Res}_{1;0}V_{-1}(\xi) \text{Res}_{1;0} \Phi_{\xi-e_1} = 0.
\end{equation*}
Since
$$
\text{Res}_{1;0}V_{1}(\xi) =-\text{Res}_{1; 0}V_{-1} (\xi) =   
(-1)^{n-1} ({\textstyle \frac{1}{2}} +g)
\prod_{1< k\leq n}  \xi_{k}^{-2}  
      \not\equiv 0,
$$
and 
$$
\text{Res}_{1;0}   \Phi_{\xi+e_1}=-\text{Res}_{1;0}   \Phi_{\xi-e_1}
 $$
by the hyperoctahedral symmetry, it follows that $ \text{Res}_{1;0}   \Phi_{\xi+e_1}(x,g) \equiv 0$ for $(\xi,x,g)\in H_{1;0}\times\mathbb{C}^n\times\mathbb{C}$, and thus
$\text{Res}_{1;2}   \Phi_{\xi} (x;g) \equiv 0$  on $H_{1;2}\times\mathbb{C}^n\times\mathbb{C}$.

If \underline{$m=3$}, then multiplication of Eq.  \eqref{diffeqn1} by $2\xi_1-1$ and performing the limit 
$2\xi_1-1\to 0$ yields that
\begin{equation*}
\bigl[ V_1(\xi) \bigr] _{2\xi_1=1} \text{Res}_{1;1}\Phi_{\xi+e_1} +
\text{Res}_{1;1} V_{-1}(\xi) \bigl[ \Phi_{\xi-e_1}-\Phi_\xi\bigr]_{2\xi_1=1} = 0  .
\end{equation*}
Since
$
\bigl[ V_1(\xi) \bigr] _{2\xi_1=1}  \not\equiv 0 ,
$
and
\begin{equation*}
\bigl[ \Phi_{\xi-e_1}\bigr]_{2\xi_1=1}=\Phi_{(-\frac12,\xi_2,\dots,\xi_n)}
=\Phi_{(\frac12,\xi_2,\dots,\xi_n)}
=\bigl[\Phi_{\xi}\bigr]_{2\xi_1=1}
\end{equation*}
by the hyperoctahedral symmetry, it follows that $ \text{Res}_{1;1}   \Phi_{\xi+e_1}(x,g) \equiv 0$ for $(\xi,x,g)\in H_{1;1}\times\mathbb{C}^n\times\mathbb{C}$, and thus
$\text{Res}_{1;3}   \Phi_{\xi} (x;g) \equiv 0$  on $H_{1;3}\times\mathbb{C}^n\times\mathbb{C}$.

If \underline{$m\ge 4$}, then multiplication of Eq.  \eqref{diffeqn1} by $2\xi_1-m+2$ and performing the limit 
$2\xi_1-m+2\to 0$ yields that
\begin{equation*}
\bigl[ V_1(\xi)\bigr]_{ 2\xi_1= m-2} \text{Res}_{1;m-2} \Phi_{\xi+e_1} = 0 .
\end{equation*}
Since
$
\bigl[ V_1(\xi) \bigr] _{2\xi_1=m-2}   \not\equiv 0 ,
$
it follows that $ \text{Res}_{1;m-2}   \Phi_{\xi+e_1}(x,g) \equiv 0$ for $(\xi,x,g)\in H_{1;m-2}\times\mathbb{C}^n\times\mathbb{C}$, and thus
$\text{Res}_{1;m}   \Phi_{\xi} (x;g) \equiv 0$  on $H_{1;m}\times\mathbb{C}^n\times\mathbb{C}$.

\subsubsection{Verification that $\emph{Res}_{12;m}\Phi_\xi(x;g)=0$ for $m\in\mathbb{Z}_{>0}$}
The required computations are similar to the previous case, but based instead on the
second difference equation of Theorem \ref{HWD:thm}  corresponding to $\ell=2$:

\begin{align}\label{diffeqn2}
\sum_{\substack{1\leq j < j^\prime\leq n \\ \epsilon,\epsilon^\prime \in \{ 1, -1\} }} 
 &V_{\{   \epsilon j,\epsilon^\prime  j^\prime \} } (\xi)  \left (\Phi_{\xi+\epsilon e_j+\epsilon^\prime e_{j^\prime}} -\Phi_\xi \right)
  + \\
  \sum_{1 \le j \le n,\epsilon\in\{1,-1\}}& U_{\{ 1,\ldots ,n\}\setminus \{ j\} , 1}(\xi) V_{\epsilon j}(\xi) \Phi_{\xi+\epsilon e_j} 
 = e^{x_1+x_2}\Phi_\xi , \nonumber
\end{align}
 where
 \begin{align*}
 V_{\{   \epsilon j,\epsilon^\prime  j^\prime \} } (\xi) =&
\frac{(\frac{1}{2}+g+\epsilon \xi_j)}{2\xi_j(2\xi_j+\epsilon)}
\frac{(\frac{1}{2}+g+\epsilon^\prime \xi_{j^\prime})}{2\xi_{j^\prime}(2\xi_{j^\prime}+\epsilon^\prime)}
\prod_{\substack{1\leq k\leq n \\k\neq j,j^\prime}} 
  (\xi_j^2-\xi_{k}^2)^{-1}  (\xi_{j^\prime}^2-\xi_{k}^2)^{-1}
\nonumber \\
&  \times
(\epsilon \xi_j+\epsilon^\prime\xi_{j^\prime})^{-1}
(1+\epsilon \xi_j+\epsilon^\prime\xi_{j^\prime})^{-1} ,
 \end{align*}
\begin{equation*}
U_{\{ 1,\ldots ,n\}\setminus \{ j\} , 1}(\xi) = - \sum_{\substack{1\leq i\leq n,\, i \neq j\\ \epsilon\in\{ 1,-1\} }}
\frac{\frac{1}{2}+g+\epsilon \xi_i}{2\xi_i(2\xi_i+\epsilon )}
 \prod_{\substack{ 1\leq k\leq n \\ k\neq j , i} }
(\xi_i^2-\xi_{k}^2)^{-1} ,
\end{equation*}
and $V_{\epsilon j}(\xi) $ is in accordance with Eq. \eqref{diffeqn1}.

 \underline{If $m=1$}, then multiplication of  Eq. \eqref{diffeqn2} by $(\xi_1+\xi_2+1)^2$ and performing the limit 
$\xi_1+\xi_2+1\to 0$ yields that
\begin{equation*}
\text{Res}_{12;-1}V_{\{ +1,+2\} }(\xi) \Bigl( \text{Res}_{12;-1}\Phi_{\xi+e_1+e_2}
-\text{Res}_{12;-1}\Phi_{\xi}  \Bigr) =0 .
\end{equation*}
Since
\begin{equation*}
\text{Res}_{12;-1} V_{\{ +1,+2\} }(\xi)
= -\Bigl[ \prod_{j\leq 2} \frac{\frac12+g+\xi_j}{2\xi_j (2\xi_j+1)} \prod_{\substack{ j\leq 2 \\ k\geq 3}} (\xi_j^2 -\xi_k^2)\Bigr]_{\xi_1+\xi_2=1}
\not\equiv 0 ,
\end{equation*}
and 
\begin{equation*}
\text{Res}_{12;-1}\Phi_{\xi+e_1+e_2} 
=-\text{Res}_{12;-1} \Phi_{\xi}
\end{equation*}
by the hyperoctahedral symmetry, it follows that $\text{Res}_{12;-1} \Phi_{\xi+e_1+e_2}(x;g)\equiv 0$ for $(\xi,x,g)\in H_{12;-1}\times\mathbb{C}^n\times \mathbb{C}$, and thus
$\text{Res}_{12;1} \Phi_{\xi}(x;g)\equiv 0$ on $H_{12;1}\times\mathbb{C}^n\times \mathbb{C}$.

If  \underline{$m=2$}, then multiplication of Eq. \eqref{diffeqn2} by  $(\xi_1+\xi_2)^2$ and performing the limit 
$\xi_1+\xi_2\to 0$ yields that
\begin{equation*}
\text{Res}_{12;0} V_{\{ +1,+2\} }(\xi) \text{Res}_{12;0 }\Phi_{\xi+e_1+e_2}
+ \text{Res}_{12;0} V_{\{ -1,-2\} }(\xi) \text{Res}_{12;0}\Phi_{\xi - e_1-e_2}
=0 .
\end{equation*}
Since 
\begin{align*}
\text{Res}_{12;0} V_{\{ +1,+2\} }(\xi)
&=- \text{Res}_{12;0} V_{\{ -1,-2\} }(\xi) \\
&=  \Bigl[ \prod_{j\leq 2} \frac{\frac12+g+\xi_j}{2\xi_j (2\xi_j+1)} \prod_{\substack{ j\leq 2 \\ k\geq 3}} (\xi_j^2 -\xi_k^2)\Bigr]_{\xi_1+\xi_2=0}
\not\equiv 0 ,
\end{align*}
and 
\begin{equation*}
\text{Res}_{12;0}\Phi_{\xi+e_1+e_2} 
=-\text{Res}_{12;0} \Phi_{\xi-e_1-e_2}
\end{equation*}
by the hyperoctahedral symmetry, it follows that
$\text{Res}_{12;0} \Phi_{\xi+e_1+e_2}(x;g)\equiv 0$ for $(\xi,x,g)\in H_{12;0}\times\mathbb{C}^n\times \mathbb{C}$, and thus
$\text{Res}_{12;2} \Phi_{\xi}(x;g)\equiv 0$ on $H_{12;2}\times\mathbb{C}^n\times \mathbb{C}$.

If  \underline{$m=3$}, then multiplication of  Eq. \eqref{diffeqn2} by $\xi_1+\xi_2-1$ and performing the limit 
$\xi_1+\xi_2\to 1$ yields that
\begin{align*}
 \bigl[ V_{\{ +1,+2\} }(\xi)  \bigr]_{\xi_1+\xi_2=1}& \text{Res}_{12;1} \Phi_{\xi+e_1+e_2} \\
&+\text{Res}_{12;1} V_{\{ -1,-2\}}(\xi) \bigl[ \Phi_{\xi-e_1 -e_2}-\Phi_\xi \bigr]_{\xi_1+\xi_2=1} 
=0 .
\end{align*}
Since $ \bigl[ V_{\{ +1,+2\} }(\xi)  \bigr]_{\xi_1+\xi_2=1}\not\equiv 0$,
and 
\begin{equation*}
\bigl[ \Phi_{\xi-e_1 -e_2}\bigr]_{\xi_1+\xi_2=1}
= \bigl[ \Phi_{\xi} \bigr]_{\xi_1+\xi_2=1} 
\end{equation*}
by the hyperoctahedral symmetry, it follows that
$\text{Res}_{12;1} \Phi_{\xi+e_1+e_2}(x;g)\equiv 0$ for $(\xi,x,g)\in H_{12;1}\times\mathbb{C}^n\times \mathbb{C}$, and thus
$\text{Res}_{12;3} \Phi_{\xi}(x;g)\equiv 0$ on $H_{12;3}\times\mathbb{C}^n\times \mathbb{C}$.

If  \underline{$m\ge 4$}, then multiplication of  Eq. \eqref{diffeqn2} by  $\xi_1+\xi_2-m+2$ and performing the limit 
$\xi_1+\xi_2\to m-2$ yields that
\begin{equation*}
\bigl[ V_{\{ +1,+2\}}(\xi)\bigr]_{\xi_1+\xi_2=m-2}  \text{Res}_{12;m-2}\Phi_{\xi+e_1+e_2}= 0.
\end{equation*}
Since $\bigl[ V_{\{ +1,+2\}}(\xi)\bigr]_{\xi_1+\xi_2=m-2}  \not\equiv 0$, it follows that
$\text{Res}_{12;m-2} \Phi_{\xi+e_1+e_2}(x;g)\equiv 0$ for $(\xi,x,g)\in H_{12;m-2}\times\mathbb{C}^n\times \mathbb{C}$, and thus
$\text{Res}_{12;m} \Phi_{\xi}(x;g)\equiv 0$ on $H_{12;m}\times\mathbb{C}^n\times \mathbb{C}$.

\subsubsection{Conclusion}
The above residue computations confirm that all singularites of $\Phi_\xi (x;g)$  \eqref{connection-formula}, \eqref{c-function} in $(\xi,x,g)$ along the hyperplanes of $\mathbb{C}^n\setminus\mathbb{C}^n_{\emph{reg}}\times\mathbb{C}^n\times\mathbb{C}$ are removable, which completes the proof of Theorem \ref{ac-wf:thm}.

\section{Proof of Proposition \ref{FW-function:prp}}\label{FW-function:prf}
The statements of the proposition are verified via
a Harish-Chandra type analysis---cf. e.g.
\cite[Ch. IV.5]{hel:groups}  and  \cite[\S 4]{has:whittaker}---which is tailored towards the current  setup of the Toda Laplacian with Morse perturbations.

\subsection{Growth estimates for $\boldsymbol{a_\nu (\xi ;g)}$}
We will first infer that---given a (any) compact subset $K\subset\mathbb{U}^n\times \mathbb{C}^n\times\mathbb{C}$---there exists a constant $C>0$ (depending only on $K$) such that
$\forall (\xi,x,g)\in K$ and $\forall \nu \geq 0$:
\begin{equation}\label{growth}
    | a_\nu (\xi ;g) e^{-\langle x,\nu\rangle } | \leq \frac{C^{\langle \nu ,\rho\rangle}}{\langle \nu , \rho\rangle !}\quad\text{where}\quad  \rho:=(n,n-1,\ldots,2,1).
\end{equation}
This bound hinges on the following (well-known) elementary estimates.
\begin{lemma}\label{estimates:lem}
There exist constants $\emph{a},\emph{b},\emph{c}>0$ such that $\forall (\xi,x,g)\in K$:
\begin{itemize}
\item[(i)]     $ | \langle \nu-2\xi ,\nu\rangle |\geq \emph{a} \langle \nu ,\rho\rangle^2 $, $\forall \nu \geq 0$; \vspace{1ex}
\item[(ii)]    $|e^{-\langle \nu ,x\rangle}|\leq \emph{b}^{\langle\nu,\rho\rangle}$, $\forall \nu \geq 0$;  \vspace{1ex}
\item[(iii)]   $ | \emph{a}_\alpha | \leq \emph{c}$, $\forall \alpha \in S$.
\end{itemize}
\end{lemma}

\begin{proof}
Since the integral cone $\{ \nu\in\mathbb{Z}^n\mid \nu\geq 0\}$ is nonnegatively generated by $S$ \eqref{S}, the angle between  $\rho$   and the vectors in this cone stays strictly sharp and bounded away from $\frac{\pi}{2}$, i.e. there exist positive constants $c_1,c_2$ such that
\begin{equation}\label{ip-est}
 \forall \nu\geq 0:\quad c_1\langle \nu ,\nu\rangle \leq    \langle \nu,\rho\rangle^2 \leq  c_2 \langle \nu ,\nu\rangle  .
\end{equation}
Let $k:= \max_{(\xi,x,g)\in K}  \langle 2\xi ,2\xi\rangle$ and pick $c_0>1$. Then for $\langle\nu ,\nu\rangle >  kc_0^2 $ the estimate in part (i) holds because of Eq. \eqref{ip-est}
and the observation that $| \langle \nu-2\xi,\nu\rangle |\geq (1-\frac{1}{c_0})\langle\nu ,\nu\rangle$, whereas for $\langle\nu ,\nu\rangle \leq  kc_0^2 $ the estimate in question follows from Eq. \eqref{ip-est} and the fact that  for $\nu >0$  the quantity $| \langle \nu-2\xi,\nu\rangle |$ stays bounded from below on $K$ away from zero.
The estimate in part (ii) follows from the Cauchy-Schwarz inequality, the compactness of $K$, and Eq. \eqref{ip-est}, while the estimate in part (iii) is plain from the compactness of $K$.
\end{proof}

To see how the asserted bound in Eq. \eqref{growth} arises from the estimates in Lemma \ref{estimates:lem}, we
pick a point $ (\xi,x,g)\in K$ and define
\begin{subequations}
\begin{equation}\label{A:def}
A_m(\xi ;g) := \max_{\substack{\nu\geq 0\\ \langle \nu ,\rho\rangle=m}} | a_\nu (\xi ; g) |
\end{equation}
for $m\in\mathbb{Z}_{\geq 0}$. The estimates in parts (i) and (iii) of Lemma \ref{estimates:lem} then imply that
\begin{equation}\label{A:bound}
A_m(\xi ;g) \leq \frac{A^m}{m!}\quad\text{with}\quad A:=1+\frac{\text{c}n}{\text{a}} .
\end{equation}
\end{subequations}
Indeed---while for $m=0$ this bound holds trivially---it is readily seen inductively in $m$ via
the recurrence relations \eqref{HCs:b}, \eqref{HCs:c}  that for $m>0$:
\begin{align*}
A_m(\xi;g) &\leq  \frac{\text{c}}{\text{a}\, m^2 } \Bigl(   n A_{m-1}(\xi;g) +A_{m-2} (\xi;g) \Bigl) \\
 &\leq  \frac{\text{c}}{\text{a}\, m^2 } \Bigl(    \frac{ nA^{m-1}}{(m-1)!}+\frac{A^{m-2}}{(m-2)!} \Bigl) \\
 &\leq \frac{A^m}{m!}      \frac{\text{c}}{\text{a} }  \Bigl( \frac{n}{A}+\frac{1}{A^2}    \Bigr) \leq \frac{A^m}{m!}  
\end{align*}
(with the implicit understanding that  the  terms involving  $A_{m-2}(\xi;g)$ and  $\frac{A^{m-2}}{(m-2)!}$ in the intermediate expressions are absent when $m=1$).
Notice at this point  that the last  of these successive estimates is immediate from the manifest inequality $A^2> \frac{\text{c}}{\text{a}} (1+nA)$.  By combining the bound in Eq. \eqref{A:bound} and the estimate in part (ii) of Lemma \ref{estimates:lem},   
the desired  upperbound  on the growth of the expansion coefficients  in Eq. \eqref{growth} follows
with $C=\text{b}A$
(which is independent of the choice  of $ (\xi,x,g)\in K$ and thus holds uniformly).

\subsection{Convergence of the Harish-Chandra series}
From the estimate in Eq. \eqref{growth} it is clear that the Harish-Chandra series $\phi_\xi(x;g)$ \eqref{HCs:a}--\eqref{HCs:c}
converges uniformly on the compact set  $K\subset\mathbb{U}^n\times \mathbb{C}^n\times\mathbb{C}$ in absolute value:
\begin{align*}
| e^{-\langle\xi ,x\rangle} \phi_\xi (x;g) | &\leq \sum_{\nu\geq 0}  | a_\nu (\xi ;g) e^{-\langle \nu ,x\rangle} | 
\leq \sum_{\nu\geq 0}  \frac{C^{\langle \nu ,\rho\rangle}}{\langle \nu , \rho\rangle !} \\
&=
\sum_{m\geq 0}  \Bigl( \sum_{\substack{ \nu\geq 0 \\ \langle \nu ,\rho\rangle =m}}    \frac{C^{\langle \nu ,\rho\rangle}}{\langle \nu , \rho\rangle !}  \Bigr) =
\sum_{m\geq 0}    \binom{n+m-1}{m}  \frac{C^m}{m !} < +\infty .
\end{align*}
Here we exploited that the number of $\nu\geq 0$ for which $\langle \nu,\rho\rangle=m$ is given explicitly  by the binomial $\binom{n+m-1}{m}$ (and thus grows at most polynomially in $m$).

\subsection{Eigenvalue equation}
The uniform convergence of the Harish-Chandra series on compacts guarantees that $\phi_\xi(x;g)$ \eqref{HCs:a}--\eqref{HCs:c}
is holomorphic for $(\xi,x,g)\in \mathbb{U}^n\times \mathbb{C}^n\times\mathbb{C}$. Substitution of the series in the eigenvalue equation \eqref{ef:eq} entails for the LHS:
\begin{equation*}
L_x \phi_\xi (x;g) = \sum_{\nu \geq 0} a_\nu (\xi;g) e^{\langle \xi -\nu,x\rangle}
\Bigl(
\langle \xi -\nu,\xi -\nu \rangle   
-\sum_{\alpha\in S} \text{a}_\alpha e^{-\langle \alpha,x\rangle}  \Bigr) ,
\end{equation*}
where we were allowed to differentiate termwise by virtue of the uniform convergence of the power series. Upon recollecting the terms and subsequently employing the recurrence relations
\eqref{HCs:b}, \eqref{HCs:c}, the  expression under consideration is rewritten as:
\begin{align*}
&\sum_{\nu \geq 0} e^{\langle \xi -\nu,x\rangle}
\Bigl(   \bigl( \langle\xi ,\xi\rangle +\langle \nu,\nu-2\xi\rangle \bigr)
  a_\nu (\xi;g) 
-\sum_{\substack{\alpha\in S\\ \nu-\alpha\geq 0}}  \text{a}_\alpha a_{\nu-\alpha} (\xi;g)  \Bigr) \\
&=
 \sum_{\nu \geq 0} \langle\xi ,\xi\rangle a_\nu (\xi;g) e^{\langle \xi -\nu,x\rangle}= \langle\xi ,\xi\rangle   \phi_\xi (x;g) ,
\end{align*}
which confirms that our fundamental Whittaker function $ \phi_\xi (x;g) $ \eqref{HCs:a}--\eqref{HCs:c} solves the asserted eigenvalue equation in Eq. \eqref{ef:eq} for the Toda Laplacian with Morse pertubations $L_x$ \eqref{L}--\eqref{A}.

\subsection{Plane-wave asymptotics}
Given a (any) wave vector $\xi\in\mathbb{U}^n$ and coupling value $g\in \mathbb{C}$, let us pick for $K\subset \mathbb{U}^n\times\mathbb{C}^n\times\mathbb{C}$ the singleton containing only the point $(\xi ,0,g)$. In this situation,  the estimate in Eq. \eqref{growth} reveals that there exists a $C>0$ (depending only on $\xi$ and $g$) such that
 $ | a_\nu (\xi;g) | \leq  \frac{C^{\langle \nu ,\rho\rangle} }{\langle \nu ,\rho\rangle !}   $  for all $\nu\geq 0$.
The plane-wave asymptotics in Eq. \eqref{ef:as} now readily follows
 by dominated convergence (assuming the additional requirement that $\text{Re}(\xi)=0$, so $| e^{\langle \xi ,x\rangle}|=1$ for all $ x\in\mathbb{R}^n$):
\begin{align*}
&\lim_{x\to +\infty} | \phi_\xi (x;g)-e^{\langle \xi ,x\rangle} |  \leq  \lim_{x\to +\infty}   \sum_{\nu > 0}  | a_\nu (\xi;g) e^{-\langle \nu,x\rangle} |
\\
& \leq   
 \lim_{x\to +\infty}   \sum_{\nu > 0}  \frac{C^{\langle \nu ,\rho\rangle}e^{-\langle \nu,x\rangle} }{\langle \nu , \rho\rangle !} 
  =     \sum_{\nu > 0} \Bigl(  \lim_{x\to +\infty} \frac{C^{\langle \nu ,\rho\rangle}e^{-\langle \nu,x\rangle} }{\langle \nu , \rho\rangle !} \Bigr) =0,
\end{align*}
where we have used that for $\nu >0$:  $e^{-\langle \nu ,x\rangle} \leq 1$ if $x_1\geq x_2\geq \cdots \geq x_n\geq 0$
and $ \lim_{x\to +\infty} e^{-\langle \nu,x\rangle}=0$.

\section{Proof of Proposition \ref{regular-domain:prp}}\label{regular-domain:prf}

To prove the proposition, we will show that our fundamental Whittaker function arises via a confluent limit---in the sense of 
Shimeno and Oshima \cite{shi:limit,osh-shi:heckman-opdam}---from the Harish-Chandra series for the hypergeometric equation with hyperoctahedral symmetry. The analytic continuation of the fundamental  Whittaker function in the spectral parameter then follows from
Opdam's corresponding result for the Harish-Chandra solution of the hypergeometric equation.

.

\subsection{Hyperoctahedral Calogero-Sutherland Laplacian}
The Hamiltonian of the hyperoctahedral  Calogero-Sutherland system is given by
\cite{ols-per:quantum,hec-sch:harmonic,opd:lecture}
\begin{align}\label{Lcs}
L_x^{\text{cs}}&= \sum_{1\leq j\leq n}
 \Bigl(  \frac{\partial^2}{\partial x_j^2}-  \frac{\frac{1}{4}k_1(k_1+2k_2-1)}{ \sinh^2 \textstyle{\frac{1}{2} } (x_j) } - \frac{k_2(k_2-1)}{ \sinh^2  (x_j)}\Bigr)  \\
&- \sum_{1\leq j<k\leq n}  \left( \frac{\frac{1}{2}k_0(k_0-1)}{ \sinh^{2} {\frac{1}{2} }   (x_j+x_k) }+ \frac{\frac{1}{2}k_0(k_0-1)}{\sinh^{2} {\frac{1}{2} }   (x_j-x_k) }\right) 
\nonumber .
\end{align}
It is well-known that Toda potentials can be interpreted as  limit degenerations of  Calogero-Sutherland potentials
\cite{sut:introduction,rui:relativistic,ino:finite,osh:completely}. Specifically,   upon setting $k_r=k_r^{(c)}$ with
\begin{subequations}
\begin{equation}\label{rescale}
k_0^{(c)}(k_0^{(c)}-1)=e^c,\quad k_1^{(c)}=2g,\quad    k^{(c)}_2(k_2^{(c)}-1)=\frac{e^{2c}}{16}\quad (\text{and}\ k_0^{(c)}, k_2^{(c)}>0),
\end{equation}
and performing a coordinate translation of the form
\begin{equation}\label{translate}
x\to x+c\rho
\end{equation}
\end{subequations}
 (with $\rho$ as in Eq. \eqref{growth}), the potential of
 the hyperoctahedral Calogero-Sutherland Laplacian $L_x^{\text{cs}}$ \eqref{Lcs} passes  in the limit $c\to +\infty$ over into
 the potential of the Toda Laplacian with Morse term $L_x$ \eqref{L}--\eqref{A}, i.e. formally  \cite{ino:finite,osh:completely}:
\begin{equation}\label{L:limit}
L_x=\lim_{c\to +\infty}  L_{x}^{\text{cs}} \left({\substack{k_r\to k_r^{(c)}\\ x\to x+c\rho}}\right)  .
\end{equation}

\subsection{Harish-Chandra series}
Upon expanding the Calogero-Sutherland potentials in the right half-plane
by means of the power series  $(e^z-e^{-z})^{-2}= \sum_{l\geq 1} l e^{-2lz}$ ($\text{Re}(z)>0$), it is readily seen that 
substitution of a formal power series of the form
\begin{subequations}
\begin{equation}\label{HC-cs:a}
\phi_\xi^{\text{cs}}(x;k_r)= \sum_{ \nu\geq 0}  a^{\text{cs}}_\nu (\xi ;k_r)  e^{\langle \xi -\nu ,x\rangle } 
\end{equation}
in the eigenvalue equation $L_x^{\text{cs}} \phi_\xi =\langle \xi, \xi\rangle \phi_\xi$, gives rise to the following recurrence relation for the expansion coefficients:
\begin{equation}\label{HC-cs:b}
 \langle \nu-2\xi ,\nu\rangle a^{\text{cs}}_\nu (\xi ;k_r) =   \sum_{\substack{\alpha\in R_+\\ l\geq 1}} l \text{a}^{\text{cs}}_\alpha (k_r) a^{\text{cs}}_{\nu-l\alpha}  (\xi ;k_r)\qquad (\nu >0) ,
\end{equation}
where $R_+:=\{ e_j , 2e_j \mid 1\leq j \leq n\} \cup \{ e_j\pm e_k \mid  1\leq j < k\leq n\}$,
\begin{equation}\label{a-alpha-cs}
\text{a}^{\text{cs}}_\alpha (k_r) :=
\begin{cases}
k_1(k_1+2k_2-1) &\text{if}\ \langle \alpha ,\alpha\rangle =1 ,\\
2k_0(k_0-1) &\text{if}\ \langle \alpha ,\alpha\rangle =2 ,\\
4k_2(k_2-1)& \text{if}\ \langle \alpha ,\alpha\rangle =4  ,
\end{cases}
\end{equation}
and we  have assumed the initial condition
\begin{equation}\label{HC-cs:d}
a^{\text{cs}}_\nu (\xi ;k_r) := \begin{cases}   1&\text{if}\ \nu =0, \\ 0&\text{if}\ \nu\not\geq 0. \end{cases} 
\end{equation}
\end{subequations}
Notice that the initial condition \eqref{HC-cs:d} guarantees in particular that the series on the RHS of the recurrence relation \eqref{HC-cs:b} amounts to a \emph{finite} sum (because  the requirement that  $\nu-l\alpha \geq 0 $  implies that $l$ remains bounded from above by $ \langle \nu ,\rho\rangle$).

The following result goes back to Opdam, cf. \cite[Cor. 2.2,\, Cor. 2.10]{opd:root}, \cite[Prp. 4.2.5]{hec-sch:harmonic}, \cite[Lem. 6.5]{opd:lecture} and \cite[Thm. 1.2]{nar-pas-pus:asymptotics}.

\begin{proposition}\label{opdam:prp} The Harish-Chandra series $\phi_\xi^{\text{cs}}(x;k_r)$ \eqref{HC-cs:a}--\eqref{HC-cs:d} constitutes an analytic function of $(\xi, x,k_r)\in \mathbb{C}^n_{+,\emph{reg}}\times\mathbb{A}^n\times\mathbb{C}^3$, where
\begin{equation}\label{chamber}
\mathbb{A}^n:= \{ x\in\mathbb{R}^n \mid x_1>x_2>\cdots >x_n>0\} .
\end{equation}
Moreover, as a (meromorphic)  function of the spectral variable $\xi\in\mathbb{C}^n$  the function under consideration has at most simple
poles along the hyperplanes belonging to
$\mathbb{C}^n\setminus \mathbb{C}^n_{+,\emph{reg}}$.
\end{proposition}

\subsection{Calogero-Sutherland $\to$ Toda confluence}
Following  \cite{shi:limit,osh-shi:heckman-opdam}, we will now lift
 the limit transition \eqref{L:limit} between the Calogero-Sutherland Laplacian $L_x^{\text{cs}}$ \eqref{Lcs} and the Toda Laplacian with Morse term $L_x$ \eqref{L}--\eqref{A} to the level of the Harish-Chandra series.

To this end let us define for a (any) bounded  domain $U\subset\mathbb{C}^n$, the normalization factor
 \begin{equation}
 \Delta_U(\xi): = \prod_{\substack{\mu>0 \\ H_\mu\cap \overline{U} \neq \emptyset}}   \langle \mu-2\xi,\mu\rangle ,
 \end{equation}
 where $H_\mu :=\{ \xi\in\mathbb{C}^n \mid  \langle \mu-2\xi,\mu\rangle =0\}$ and $\overline{U}$ refers to the (compact) closure of $U$.
 It is clear from the recurrence relations in Eqs. \eqref{HC-cs:b}--\eqref{HC-cs:d}  that  for any $\nu\geq 0$
 the normalized Harish-Chandra coefficient $\Delta_U(\xi) a^{\text{cs}}_\nu (\xi;k_r)$  is holomorphic and bounded in $\xi$ on the bounded open connected set $U$.

\begin{proposition}[Confluent limit of the Harish-Chandra Series]\label{HCcs-lim:prp}
For any $(\xi,x,g) \in U\times\mathbb{A}^n\times\mathbb{C}$, one has that
\begin{equation}
\lim_{c\to +\infty}   e^{-c\langle \xi ,\rho\rangle} \Delta_U(\xi) \phi_\xi^{\emph{cs}}(x+c\rho;k_r^{(c)})= \Delta_U(\xi) \phi_\xi (x;g) .
\end{equation}
\end{proposition}

\begin{proof}
It is immediate from the Harish-Chandra series \eqref{HC-cs:a}--\eqref{HC-cs:d} that
\begin{equation}\label{series}
 e^{-c\langle \xi ,\rho\rangle} \Delta_U(\xi) \phi_\xi^{\text{cs}}(x+c\rho;k_r^{(c)})= \sum_{ \nu\geq 0}  \hat{a}^{\text{cs}}_\nu (\xi ;k_r^{(c)})  e^{\langle \xi -\nu ,x\rangle } ,
\end{equation}
where $\hat{a}^{\text{cs}}_\nu (\xi ;k_r^{(c)})$, $\nu\geq 0$ is determined by the rescaled recurrence
\begin{subequations}
\begin{equation}\label{rec:a}
 \langle \nu-2\xi ,\nu\rangle \hat{a}^{\text{cs}}_\nu (\xi ;k_r^{(c)}) =   \sum_{\substack{\alpha\in R_+\\ l\geq 1}} l e^{-c\, l \langle \alpha ,\rho\rangle } \text{a}^{\text{cs}}_\alpha (k_r^{(c)}) \hat{a}^{\text{cs}}_{\nu-l\alpha}  (\xi ;k_r^{(c)})\qquad (\nu >0) ,
\end{equation}
with the modified initial condition
\begin{equation}\label{rec:b}
\hat{a}^{\text{cs}}_\nu (\xi ;k_r^{(c)}) := \begin{cases} \Delta_U(\xi)  &\text{if}\ \nu =0, \\ 0&\text{if}\ \nu\not\geq 0. \end{cases} 
\end{equation}
\end{subequations}
Since in the finite sum on the RHS of the recurrence \eqref{rec:a}
\begin{equation}\label{a-cs-t-lim}
\lim_{c\to +\infty}  e^{-c\, l \langle \alpha ,\rho\rangle } \text{a}^{\text{cs}}_\alpha (k_r^{(c)})=\begin{cases}
\text{a}_\alpha & \text{if}\  \alpha\in S\ \text{and}\ l=1 ,\\
0& \text{otherwise},
\end{cases}
\end{equation}
we
recover for $c\to +\infty$ the recurrence in Eq. \eqref{HCs:b} with the initial condition \eqref{HCs:c} multiplied by $\Delta_U(\xi)$, i.e.
\begin{equation}\label{coef-lim}
\lim_{c\to +\infty}   \hat{a}^{\text{cs}}_\nu (\xi ;k_r^{(c)}) =\Delta_U(\xi) a_\nu (\xi ;g)  \quad (\forall \xi\in U,\nu \geq 0) .
\end{equation}
Moreover,  by a standard argument  (cf. Subsection \ref{HC-cs:prf} below for the precise details in the present setting) the above recurrence relations imply
that  for any $\varepsilon >0$ there exists a constant $A>0$ (depending only on  $\xi\in U$, $g\in\mathbb{C}$ and $\varepsilon$)  such that
\begin{equation}\label{HC-cs:estimate}
\forall  c\geq 0, \nu\geq 0: \quad | \hat{a}^{\text{cs}}_\nu (\xi ;k_r^{(c)})| \leq A\,  e^{\varepsilon \langle \nu ,\rho\rangle} .
\end{equation}
Upon picking $\varepsilon >0$ sufficiently small such that $x\in \varepsilon\rho +\mathbb{A}^n$, it is seen that
 the series on the RHS of Eq. \eqref{series} can be bounded term wise and uniformly in $c$ by a  absolutely convergent series:
\begin{equation}
 \sum_{ \nu\geq 0}  | \hat{a}^{\text{cs}}_\nu (\xi ;k_r^{(c)})  e^{\langle \xi -\nu ,x  \rangle } |
\leq A   | e^{\langle \xi ,x\rangle } | \sum_{ \nu\geq 0}   e^{- \langle \nu ,x-\varepsilon\rho \rangle }< +\infty .
\end{equation}
The usphot is that the asserted limit transition follows from Eqs. \eqref{series} and \eqref{coef-lim} by dominated convergence.
\end{proof}

Since the multiplication by $\Delta_U(\xi)$ regularizes the coefficients of the Harish-Chandra series \eqref{HCs:a}--\eqref{HCs:c} for $\xi\in U$, it is clear from its local absolute uniform convergence (cf. Section \ref{FW-function:prf}) that  the normalized Harish-Chandra series  $\Delta_U(\xi) \phi_\xi (x;g)$ constitutes an analytic function of $(\xi ,x,g)\in U\times \mathbb{C}^n\times\mathbb{C}$.
One concludes from Propositions \ref{opdam:prp} and \ref{HCcs-lim:prp} that for $(\xi,x,g)\in (U\cap H_\mu)\times\mathbb{A}^n\times\mathbb{C}$ ($\mu >0$)
 the function in question vanishes---and is thus divisible by $\langle \mu -2\xi,\mu\rangle$---except (possibly) when the hyperplane  $H_\mu$ belongs to
 $\mathbb{C}^n\setminus \mathbb{C}^n_{+,\text{reg}}$. But then, by analytic continuation, it is clear that in fact $\Delta_U(\xi) \phi_\xi (x;g)=0$ for
 $(\xi,x,g)\in
 (U\cap H_\mu)\times \mathbb{C}^n\times\mathbb{C}$  (if $H_\mu$, $\mu>0$ does not belong to  $\mathbb{C}^n\setminus \mathbb{C}^n_{+,\text{reg}}$).

 Since the bounded domain $U\subset\mathbb{C}^n$ was chosen arbitrarily, this means that  the fundamental Whittaker function $\phi_\xi (x;g)$ is regular in $\xi$ along the hyperplane $H_\nu$ for all
 $\nu >0$ such that $H_\nu\not\subset \mathbb{C}^n\setminus \mathbb{C}^n_{+,\text{reg}}$, and that it possesses at most a simple pole along the hyperplane $H_\nu$ for all  $\nu >0$ such that $H_\nu\subset \mathbb{C}^n\setminus \mathbb{C}^n_{+,\text{reg}}$, which
completes the proof of Proposition \ref{regular-domain:prp}. 

\subsection{Proof of the bound in Eq. \eqref{HC-cs:estimate}}\label{HC-cs:prf}
After fixing $\xi\in U$ and $g\in\mathbb{C}$, we follow the proof of  a similar bound in \cite[Ch. IV, Lem. 5.3]{hel:groups}. 
To this end let us pick $\text{a}>0$ and $N>0$ such that:
\begin{equation}
\langle \nu ,\nu -2\xi\rangle \geq \text{a}\langle \nu ,\rho\rangle^2\qquad \forall\nu \geq 0\ \text{with}\  \langle \nu,\rho\rangle \geq N
\end{equation}
(cf. Lemma \ref{estimates:lem}).
In view of  Eqs. \eqref{rescale} and \eqref{a-alpha-cs}, it is clear that
for a given $\varepsilon >0$ there exists a $C>0$ such that $\forall c\geq 0$:
\begin{equation}
\frac{1}{\text{a}} 
 \sum_{\substack{\alpha\in R_+\\ l\geq 1}}  e^{-(\varepsilon+c) l \langle \alpha ,\rho\rangle } |\text{a}^{\text{cs}}_\alpha (k_r^{(c)}) |
 \leq C.
\end{equation}
Upon fixing a positive integer $M\geq\max (N,C)$, we now pick $A>0$ sufficiently large such that $\forall c\geq 0$:
$| \hat{a}^{\text{cs}}_\nu (\xi ;k_r^{(c)}) |\leq A e^{\varepsilon \langle \nu ,\rho\rangle}$ for all $\nu\geq 0$ with $\langle\nu,\rho\rangle< M$. (The existence of  $A$ is clear from the limit in Eq. \eqref{coef-lim} and the continuity in $c$.)

One then sees inductively from the recurrence in Eq. \eqref{rec:a} that also for all $\nu\geq 0$ such that $\langle\nu,\rho\rangle\geq M$:
\begin{align*}
| \hat{a}^{\text{cs}}_\nu (\xi ;k_r^{(c)}) | &\leq
\frac{1}{\text{a} \langle \nu ,\rho\rangle^2}
  \sum_{\substack{\alpha\in R_+\\ l\geq 1}}  l e^{-c\, l \langle \alpha ,\rho\rangle } | \text{a}^{\text{cs}}_\alpha (k_r^{(c)})  \hat{a}^{\text{cs}}_{\nu -l\alpha}(\xi ;k_r^{(c)}) |   \\
&\leq
\frac{1}{\text{a} \langle \nu ,\rho\rangle}
  \sum_{\substack{\alpha\in R_+\\ l\geq 1}}  e^{-c\, l \langle \alpha ,\rho\rangle } | \text{a}^{\text{cs}}_\alpha (k_r^{(c)}) |   A e^{\varepsilon \langle \nu-l\alpha ,\rho\rangle}
  \\
& \leq {\textstyle \frac{M}{ \langle \nu ,\rho\rangle}} A e^{\varepsilon \langle \nu ,\rho\rangle} \leq A e^{\varepsilon \langle \nu ,\rho\rangle}
\end{align*}
(where we exploited that  $l\leq \langle \nu ,\rho\rangle$ if $\nu-l\alpha \geq 0 $).

\section{Proof of Theorem \ref{HWD:thm}}\label{HWD:prf}
In this section we  retrieve our difference equations for the hyperoctahedral Whittaker function from the hypergeometric difference equations
in \cite[Thm. 2]{die-ems:difference},
with the aid of the Calogero-Sutherland $\to$ Toda confluence from Section \ref{regular-domain:prf}.

\subsection{Hyperoctahedral hypergeometric function}
For $(\xi,x,k_r)\in\mathbb{C}_{\text{reg}}^n\times \mathbb{A}^n\times\mathbb{C}^3$, let us consider the following eigenfunction
of the Calogero-Sutherland Laplacian $L_x^{\text{cs}}$ \eqref{Lcs}:
\begin{subequations}
\begin{equation}\label{connection-formula-cs}
\Phi_\xi^{\text{cs}}(x;k_r):=\sum_{w\in W}  C^{\text{cs}}(w \xi;k_r)  \phi_{w\xi}^{\text{cs}}(x;k_r) ,
\end{equation}
where
\begin{align} \label{c-function-cs}
C^{\text{cs}}(\xi;k_r):=&
 \prod_{1\leq j\leq n}    \frac{\Gamma (2\xi_j) \Gamma (\frac{1}{2}k_1+\xi_j)}{\Gamma (k_1+2\xi_j) \Gamma (\frac{1}{2}k_1+k_2+\xi_j)} \\
& \times \prod_{1\leq j<k\leq n}   \frac{\Gamma (\xi_j+\xi_k ) \Gamma (\xi_j-\xi_k)}{\Gamma (k_0+\xi_j+\xi_k)\Gamma (k_0+\xi_j-\xi_k)} . \nonumber
\end{align}
\end{subequations}
Up to an overal normalization factor---depending on $x$ and $k_r$ but \emph{not} on $\xi$---this is the hyperoctahedral hypergeometric function of Heckman and Opdam
associated with the root system $BC_n$ \cite{opd:root,hec-sch:harmonic,opd:lecture} (cf. Remark \ref{HO:rem} below).

\begin{lemma}\label{Ccs-lim:lem}
For any $(\xi,g)\in\mathbb{C}_{\emph{reg},-}^n\times\mathbb{C}$ (cf. Eq. \eqref{C-}),  one has that
\begin{subequations}
\begin{equation}
\lim_{c\to +\infty}   \gamma (k_r^{(c)}) \,e^{c\langle \xi ,\rho\rangle }  C^{\emph{cs}}(\xi;k_r^{(c)}) =C (\xi;g) ,
\end{equation}
where
\begin{equation}
\gamma (k_r):=  \left( \Gamma (k_0)    \right)^{n(n-1)}  \left(  \frac{\Gamma (k_1)\Gamma (\frac{1}{2}k_1+k_2)}{\Gamma (\frac{1}{2} k_1)\Gamma (\frac{1}{2}+\frac{1}{2}k_1)} \right)^n .
\end{equation}
\end{subequations}
\end{lemma}
\begin{proof}
This lemma is immediate from the   duplication formula for the gamma function
 \cite[5.5.5]{olv-loz-boi-cla:nist}
 \begin{equation*}
 \Gamma (2z)=\pi^{-\frac{1}{2}} 2^{2z-1}\Gamma (z)\Gamma \textstyle{(\frac{1}{2}}+z)
 \end{equation*}
($2z\not\in\mathbb{Z}_{\leq 0}$) and the asymptotics
 \cite[5.11.12]{olv-loz-boi-cla:nist}
 \begin{equation*}
 \lim_{c\to +\infty}   \frac{c^z \Gamma (c)  }{\Gamma (z+c)} =1.
 \end{equation*}
 \end{proof}
 The upshot is that the confluence limit for the Harish-Chandra series in  Proposition \ref{HCcs-lim:prp} persists at the level of the connection formulas.
 
 \begin{proposition}[Confluent limit of the Connection Formula]\label{CFcs-lim:prp}
For any $(\xi,x,g)\in \mathbb{C}^n_{\emph{reg}}\times\mathbb{A}^n\times\mathbb{C}$, one has that
\begin{equation}\label{CFcs-lim:eq}
\lim_{c\to +\infty} \gamma (k_r^{(c)})  \Phi_\xi^{\emph{cs}}(x+c\rho;k_r^{(c)})= \Phi_\xi (x;g) .
\end{equation}
\end{proposition}
\begin{proof}
It is clear from  Proposition \ref{HCcs-lim:prp}, Lemma \ref{Ccs-lim:lem} and the connection formula in
Eqs \eqref{connection-formula-cs}, \eqref{c-function-cs}, that the limit on the LHS of Eq. \eqref{CFcs-lim:eq} reproduces the connection formula
for the hyperoctahedral  Whittaker function in Eqs \eqref{connection-formula}, \eqref{c-function}.
\end{proof}

For reduced root systems such confluences from the hypergeometric function of Heckman and Opdam to the class-one Whittaker function studied by Hashizume were established by Shimeno \cite{shi:limit}. In the case of a single variable ($n=1$), the limit in Proposition \ref{CFcs-lim:prp} was detailed in \cite[\S 2]{osh-shi:heckman-opdam}.

\begin{remark}\label{HO:rem}
The precise relation between $\Phi_\xi^{\text{cs}}(x;k_r) $ \eqref{connection-formula-cs}, \eqref{c-function-cs} and the hyperoctahedral hypergeometric function
$ F_{BC_n}(\xi,x;k_r)$ of Heckman and Opdam  \cite{opd:root,hec-sch:harmonic,opd:lecture}  is given by:
\begin{subequations}
\begin{equation}
\Phi_\xi^{\text{cs}}(x;k_r) =  { \delta (x;k_r)}{C^{\text{cs}}(\rho(k_r) ;k_r)}     F_{BC_n}(\xi,x;k_r)
\end{equation}
($(\xi,x)\in\mathbb{C}_{\text{reg}}^n\times \mathbb{A}^n$), where
\begin{align}
\delta (x;k_r) :=& \prod_{1\leq j\leq n}   (e^{\frac{1}{2} x_j}  -  e^{-\frac{1}{2} x_j} )^{k_1} (e^{x_j}  -  e^{-x_j} )^{k_2}  \\
&\times \prod_{1\leq j <k\leq n}      (e^{\frac{1}{2} (x_j+x_k)}  -  e^{-\frac{1}{2} (x_j+x_k)} )^{k_0}   (e^{\frac{1}{2} (x_j-x_k)}  -  e^{-\frac{1}{2} (x_j-x_k)} )^{k_0}
\nonumber
\end{align}
and
\begin{equation}
\rho (k_r):={\textstyle ((n-1)k_0+\frac{1}{2}k_1+k_2,(n-2)k_0+\frac{1}{2}k_1+k_2,\ldots,\frac{1}{2}k_1+k_2)} .
\end{equation}
\end{subequations}
\end{remark}

\subsection{Hypergeometric difference equations}
In view of Remark \ref{HO:rem}, it follows
from \cite[Thm. 2]{die-ems:difference}  that
for any $x\in\mathbb{A}^n$ \eqref{chamber}, $k_r\in\mathbb{C}$ ($r=0,1,2$) and $\ell \in \{ 1,\ldots ,n\}$,
the normalized hypergeometric function $\Phi^{\text{cs}}_\xi(x;k_r)$ \eqref{connection-formula-cs}, \eqref{c-function-cs} satisfies the difference equation\footnote{For $\ell=1$, this hypergeometric difference equation had been found before by Chalykh, cf.  \cite[Thm 6.12]{cha:bispectrality}.}

{\small
\begin{equation*}
\sum_{\substack{J\subset \{ 1,\ldots ,n\} ,\, 0\leq|J|\leq \ell\\
               \varepsilon_j=\pm 1,\; j\in J}}
\!\!\!\!\!\!\!\!\!
U^{\text{cs}}_{J^c,\, \ell -|J|}(\xi ;k_r)
V^{\text{cs}}_{\varepsilon J}(\xi ;k_r)
\Phi^{\text{cs}}_{\xi +e_{\varepsilon J}} (x;k_r) =E_\ell (x) \Phi^{\text{cs}}_{\xi} (x;k_r)
\end{equation*}}
(as an identity between meromorphic functions of $\xi \in \mathbb{C}^n$), where
{\small
\begin{align*}
V^{\text{cs}}_{\varepsilon J}(\xi ;k_r)&:=
\prod_{j\in J} 
\frac{(\varepsilon_j\xi_j+\frac{1}{2}k_1+k_2)(1+2\varepsilon_j\xi_j+k_1)}{\varepsilon_j\xi_j(1+2\varepsilon_j\xi_j)}
\prod_{\substack{j\in J\\ k\not\in J}} 
\Bigl(\frac{\varepsilon_j\xi_j+\xi_{k}+k_0}{\varepsilon_j\xi_j+\xi_{k}}\Bigr)\Bigl(\frac{\varepsilon_j\xi_j-\xi_{k}+k_0}{\varepsilon_j\xi_j-\xi_{k}}\Bigr)
\nonumber \\
&  \times
\prod_{\substack{j,j^\prime \in J\\ j<j^\prime}}
\Bigl(\frac{\varepsilon_j\xi_j+\varepsilon_{j^\prime}\xi_{j^\prime}+k_0}{\varepsilon_j\xi_j+\varepsilon_{j^\prime}\xi_{j^\prime}}\Bigr)
\Bigl(\frac{1+\varepsilon_j\xi_j+\varepsilon_{j^\prime}\xi_{j^\prime}+k_0}{1+\varepsilon_j\xi_j+\varepsilon_{j^\prime}\xi_{j^\prime}}\Bigr) ,
\end{align*}
\begin{align*}
U^{\text{cs}}_{K,p}(\xi ;k_r):=
 (-1)^p
\sum_{\stackrel{I\subset K,\, |I|=p}
               {\varepsilon_i =\pm 1,\; i\in I }}
&\Biggl( \prod_{i\in I} 
\frac{(\varepsilon_i\xi_i+\frac{1}{2}k_1+k_2)(1+2\varepsilon_i\xi_i+k_1)}{\varepsilon_i\xi_i(1+2\varepsilon_i\xi_i)} \nonumber \\
&\times \prod_{\substack{i\in I\\ k\in K\setminus I}} 
\Bigl(\frac{\varepsilon_i\xi_i+\xi_{k}+k_0}{\varepsilon_i\xi_i+\xi_{k}}\Bigr)\Bigl(\frac{\varepsilon_i\xi_i-\xi_{k}+k_0}{\varepsilon_i\xi_i-\xi_{k}}\Bigr)
\nonumber \\
&  \times
\prod_{\substack{i,i^\prime \in I\\ i<i^\prime}}
\Bigl(\frac{\varepsilon_i\xi_i+\varepsilon_{i^\prime}\xi_{i^\prime}+k_0}{\varepsilon_i\xi_i+\varepsilon_{i^\prime}\xi_{i^\prime}}\Bigr)
\Bigl(\frac{1+\varepsilon_i\xi_i+\varepsilon_{i^\prime}\xi_{i^\prime}-k_0}{1+\varepsilon_i\xi_i+\varepsilon_{i^\prime}\xi_{i^\prime}}\Bigr) \Biggr) ,\end{align*}}
 and
{\small
\begin{equation*}
E_\ell (x):=4^{\ell }\sum_{\substack{J\subset \{ 1,\ldots, n\} \\ |J|=\ell}}   \prod_{j\in J}  \sinh^2\left(\frac{x_j}{2}\right) .
\end{equation*}}

If we now pick $(\xi ,x,g)\in \mathbb{C}^n_{\text{reg}}\times \mathbb{A}^n\times\mathbb{C}$ and
perform the substitution \eqref{rescale}, \eqref{translate} into the above hypergeometric identities,
then---after multiplication of both sides by  an overall scaling factor of the form $e^{-\frac{c}{2}\ell (2n+1-\ell)} \gamma (k_r^{(c)})$---the difference equations in Theorem \ref{HWD:thm} are recovered in the limit $c\to +\infty$. Indeed, this is immediate from
Proposition \ref{CFcs-lim:prp} and the elementary limits
\begin{align*}
\lim_{c\to +\infty}  e^{-\frac{c}{2}|J| (2n+1-|J|)} V^{\text{cs}}_{\varepsilon J}(\xi ;k_r^{(c)})&=V_{\varepsilon J}(\xi ), \\
\lim_{c\to +\infty}  e^{-\frac{c}{2} p (2|K|+1-p)}  U^{\text{cs}}_{K,p}(\xi ;k_r^{(c)}) &= U_{K,p}(\xi) ,\\
\lim_{c\to +\infty}  e^{-\frac{c}{2}\ell (2n+1-\ell)}   E_\ell (x+c\rho)&=  e^{x_1+\cdots +x_\ell},
\end{align*}
upon noting that $|J| (2n+1-|J|)+p (2|K|+1-p)=\ell (2n+1-\ell)$ when $|K|=n-|J|$ and $p=\ell-|J|$.
This completes the proof of Theorem \ref{HWD:thm} for $x\in\mathbb{A}^n$ \eqref{chamber}. The extension to arbitrary $x\in\mathbb{C}^n$ is plain by analytic continuation, as
$\Phi_\xi (x;g)$ \eqref{connection-formula}, \eqref{c-function} constitutes an entire function of $x$ (cf. Propositions \ref{FW-function:prp} and \ref{regular-domain:prp}).

\bibliographystyle{amsplain}

\begin{thebibliography}{000000}

\bibitem[Ba]{bab:equations} O. Babelon,
Equations in dual variables for Whittaker functions,
Lett. Math. Phys. {\bf 65} (2003), 229--240. 

\bibitem[BO]{bau-oco:exponential} F. Baudoin and N. O'Connell,
Exponential functionals of Brownian motion and class-one Whittaker functions,
Ann. Inst. Henri Poincar\'e Probab. Stat. {\bf 47} (2011), 1096--1120. 

\bibitem[Bo]{bog:perturbations} O.I. Bogoyavlensky,
On perturbations of the periodic Toda lattice,
Comm. Math. Phys. {\bf 51} (1976), 201--209. 

\bibitem[BC]{bor-cor:macdonald} A. Borodin and I. Corwin, Macdonald processes, Probab. Theory Related Fields {\bf 158} (2014), 225--400.

\bibitem[BBL]{bru-bum-lic:whittaker} B. Brubaker, D. Bump, and A. Licata, 
Whittaker functions and Demazure operators,
J. Number Theory {\bf 146} (2015), 41--68. 

\bibitem[C]{cha:bispectrality} O.A. Chalykh,
Bispectrality for the quantum Ruijsenaars model and its integrable deformation,
J. Math. Phys. {\bf 41} (2000), 5139--5167. 

\bibitem[COSZ]{cor-oco-sep-zyg:tropical} I. Corwin, N. O'Connell, T. Sepp\"al\"ainen, and N. Zygouras, Tropical combinatorics and Whittaker functions, Duke Math. J. {\bf 163} (2014), 513--563. 

\bibitem[D1]{die:integrability} J.F. van Diejen,  Integrability of difference Calogero-Moser systems, J. Math. Phys. {\bf 35} (1994), 2983--3004.

\bibitem[D2]{die:difference}  \bysame, Difference Calogero-Moser systems and finite Toda chains, J. Math. Phys. {\bf 36} (1995), 1299--1323. 

\bibitem[DE]{die-ems:difference}  J.F. van Diejen and E. Emsiz, Difference equation for the Heckman-Opdam hypergeometric function and its confluent Whittaker limit, Adv. Math. {\bf 285} (2015), 1225--1240.

\bibitem[DG]{dui-gru:differential} J.J. Duistermaat and F.A. Gr\"unbaum, Differential equations in the spectral parameter, Comm. Math. Phys. {\bf 103} (1986),  177--240.

\bibitem[E]{eti:whittaker} P. Etingof, Whittaker functions on quantum groups and $q$-deformed Toda operators, in: Differential Topology, Infinite-dimensional Lie Algebras, and Applications (A. Astashkevich and S. Tabachnikov, eds.),  Amer. Math. Soc. Transl. Ser. 2, Vol. 194, Amer. Math. Soc., Providence, RI, 1999, p. 9--25.

\bibitem[F]{feh:action} L. Feh\'er,
Action-angle map and duality for the open Toda lattice in the perspective of Hamiltonian reduction,
Phys. Lett. A {\bf 377} (2013), 2917--2921. 

\bibitem[GLO]{ger-leb-obl:new}  A. Gerasimov, D. Lebedev, and S. Oblezin, 
New integral representations of Whittaker functions for classical Lie groups,  Russian Math. Surveys {\bf 67} (2012), 1--92. 

\bibitem[GW]{goo-wal:classical} R. Goodman and N.R. Wallach, Classical and quantum mechanical systems 
of Toda-Lattice type III. Joint eigenfunctions of the quantized systems, Commun. Math. Phys. {\bf 105} (1986), 473--509.

\bibitem[G]{gru:bispectral} F.A. Gr\"unbaum,
The bispectral problem: an overview, in: Special Functions 2000: Current Perspective and Future Directions,
J. Bustoz, M.E.H. Ismail, and S.K. Suslov (eds.),
NATO Sci. Ser. II Math. Phys. Chem., Vol.  {30}, Kluwer Acad. Publ., Dordrecht, 2001, p. 129--140.

\bibitem[HR]{hal-rui:kernel} M. Halln\"as and S.N.M. Ruijsenaars, Kernel functions and B\"acklund transformations for relativistic Calogero-Moser and Toda systems, J. Math. Phys. {\bf 53} (2012), 123512.

\bibitem[H]{has:whittaker} M. Hashizume,
Whittaker functions on semisimple Lie groups,
Hiroshima Math. J. {\bf 12} (1982),  259--293. 

\bibitem[HS]{hec-sch:harmonic} G. Heckman and H. Schlichtkrull,
Harmonic Analysis and Special Functions on Symmetric Spaces,
Perspectives in Mathematics, Vol. 16, Academic Press, Inc., San Diego, CA, 1994.

\bibitem[He]{hel:groups} S. Helgason,
Groups and Geometric Analysis: 
Integral Geometry, Invariant Differential Operators, and Spherical Functions, Mathematical Surveys and Monographs, Vol. 83, American Mathematical Society, Providence, RI, 2000.

\bibitem[I]{ino:finite} V.I. Inozemtsev, The finite Toda lattices, Comm. Math. Phys. {\bf 121} (1989), 629--638.

 \bibitem[ISh]{ior-sha:wave} N.Z. Iorgov and V.N. Shadura, Wave functions of the Toda chain with boundary interaction, Theoret. and Math. Phys. {\bf 142} (2005), 289--305.
 
\bibitem[ISt]{ish-sta:new} T. Ishii and E. Stade, 
New formulas for Whittaker functions on $GL(n,\mathbb{R})$,
J. Funct. Anal. {\bf 244} (2007), 289--314. 

\bibitem[J]{jac:fonctions} H. Jacquet,
Fonctions de Whittaker associ\'ees aux groupes de Chevalley,
Bull. Soc. Math. France {\bf 95} (1967), 243--309. 

\bibitem[KL]{kha-leb:integral} S. Kharchev and D. Lebedev,
Integral representations for the eigenfunctions of quantum open and periodic Toda chains from the QISM formalism,
J. Phys. A {\bf 34} (2001), 2247--2258. 

\bibitem[K]{kos:quantization} B. Kostant, Quantization and representation theory, in:
Representation Theory of Lie groups (G.L. Luke, ed.), London Mathematical Society Lecture Note Series, Vol. 34, Cambridge University Press, Cambridge-New York, 1979, p. 287--316.

\bibitem[Ko]{koz:aspects}  K.K. Kozlowski, Aspects of the inverse problem for the Toda chain, J. Math. Phys. {\bf 54} (2013), 121902.

\bibitem[L]{lag:schrodinger} J.C. Lagarias,  The Schr\"odinger operator with Morse potential on the right half-line,
Commun. Number Theory Phys. {\bf 3} (2009), 323--361.

\bibitem[M]{mac:symmetric}  I.G. Macdonald, {\em Symmetric Functions and
Hall Polynomials}, Second Edition, Clarendon Press, Oxford, 1995.

\bibitem[NPP]{nar-pas-pus:asymptotics} E.K. Narayanan, A. Pasquale, and S. Pusti, Asymptotics of Harish-Chandra expansions, bounded hypergeometric functions associated with root systems, and applications, Adv. Math. {\bf 252} (2014), 227--259.

\bibitem[OSZ]{oco-sep-zyg:geometric}  N. O'Connell, T. Sepp\"al\"ainen, and N. Zygouras, Geometric RSK correspondence, Whittaker functions and symmetrized random polymers, Invent. Math. {\bf 197} (2014), 361--416.

\bibitem[OP]{ols-per:quantum} M.A. Olshanetsky and A.M. Perelomov,
Quantum integrable systems related to Lie algebras, 
Phys. Rep. {\bf 94} (1983), 313--404. 

\bibitem[OLBC]{olv-loz-boi-cla:nist} F.W.J. Olver, D.W. Lozier, R.F. Boisvert and C.W. Clark. (eds.),
NIST Handbook of Mathematical Functions,
Cambridge University Press, Cambridge, 2010.

\bibitem[O1]{opd:root} E.M. Opdam, Root systems and hypergeometric functions. IV, Compos. Math. {\bf 67} (1988), 191--209.

\bibitem[O2]{opd:lecture} \bysame,
Lecture Notes on Dunkl Operators for Real and Complex Reflection Groups,
MSJ Memoirs, Vol. 8, Mathematical Society of Japan, Tokyo, 2000. 

\bibitem[Os]{osh:completely} T. Oshima, Completely integrable systems associated with classical root systems, SIGMA Symmetry Integrability Geom. Methods Appl. {\bf 3} (2007), Paper 061, 50 pp. 

\bibitem[OS]{osh-shi:heckman-opdam} T. Oshima and N. Shimeno,
Heckman-Opdam hypergeometric functions and their specializations, in: New Viewpoints of Representation Theory and Noncommutative Harmonic Analysis, M. Itoh and H. Ochiai (eds.),
RIMS K\^{o}ky\^{u}roku Bessatsu, Vol. B20, Res. Inst. Math. Sci. (RIMS), Kyoto, 2010, 129--162, 

\bibitem[Ri]{rie:mirror}  K. Rietsch, A mirror symmetric solution to the quantum Toda lattice, Comm. Math. Phys. {\bf 309} (2012),  23--49. 

\bibitem[R1]{rui:complete} S.N.M. Ruijsenaars, Complete integrability of relativistic Calogero-Moser systems and elliptic function identities. Comm. Math. Phys. {\bf 110} (1987), 191--213. 

\bibitem[R2]{rui:relativistic} \bysame, Relativistic Toda systems, Comm. Math. Phys. {\bf 133} (1990), 217--247. 

\bibitem[Se]{sem:quantisation} M. Semenov-Tian-Shansky, Quantisation of open Toda lattices, in: Dynamical Systems VII: Integrable Systems, Nonholonomic Dynamical Systems, V.I. Arnol'd and S.P. Novikov (eds.),
Encyclopaedia of Mathematical Sciences, Vol. 16, Springer-Verlag, Berlin, 1994, 226--259.

\bibitem[Sh]{shi:limit} N. Shimeno, A limit transition from Heckman-Opdam hypergeometric functions to the Whittaker functions associated with root systems, arXiv.0812.3773.

\bibitem[Sk1]{skl:boundary} E.K. Sklyanin,
Boundary conditions for integrable quantum systems,
J. Phys. A {\bf 21} (1988), 2375–--2389.

\bibitem[Sk2]{skl:bispectrality} \bysame, Bispectrality for the quantum open Toda chain,
J. Phys. A {\bf 46} (2013), 382001.

\bibitem[St1]{sta:explicit} E. Stade, On explicit integral formulas for $GL(n,\mathbb{R})$-Whittaker functions, Duke Math. J. {\bf 60} (1990), 313--362. 

\bibitem[St2]{sta:mellin} \bysame,
Mellin transforms of $GL(n,\mathbb{R})$ Whittaker functions,
Amer. J. Math. {\bf 123} (2001), 121--161. 

\bibitem[Su]{sut:introduction} B. Sutherland, An introduction to the Bethe ansatz, in: Exactly Solvable Problems in Condensed Matter and Relativistic Field Theory (Panchgani, 1985), B.S. Shastry, S.S. Jha and V. Singh (eds.),
Lecture Notes in Phys., Vol. {242}, Springer Verlag, Berlin, 1985, p. 1--95.


\end{thebibliography}

\end{document}